\documentclass[11pt]{article}
\usepackage{latexsym,amsmath,url,epsfig}
\usepackage{tikz}
\usepackage{textcomp}
\usepackage{amsfonts,euscript}
\usepackage{amsmath}
\usepackage{amssymb}
\usepackage{amsfonts}
\usepackage{amsthm}
\usepackage{mathrsfs}
\usepackage[all]{xy}
\usepackage{epstopdf}
\usepackage{subfigure}
\usepackage{rotating}
\usepackage[title]{appendix}
\usepackage{appendix}
\usepackage{multirow,rotating}
\usepackage{url}
\usepackage{algorithm}
\usepackage{algorithmicx}
\usepackage{algpseudocode}
\usepackage{subfigure}
\usepackage{setspace}
\usepackage{listings}
\usepackage{color}
\usepackage{tikz}
\usepackage{graphics}
\usepackage[breaklinks=true]{hyperref}
\usepackage[english]{babel}
\usepackage{datetime}
\usepackage{subfigure}
\usepackage{booktabs}
\usepackage{rotfloat}
\usepackage{framed}
\usepackage[authoryear]{natbib}
\usepackage{graphicx}

\usetikzlibrary{chains,shapes,decorations,arrows,calc,arrows.meta,fit,positioning}
\tikzset{
    -Latex,auto,node distance =1 cm and 1 cm,semithick,
    state/.style ={ellipse, draw, minimum width = 0.7 cm},
    point/.style = {circle, draw, inner sep=0.04cm,fill,node contents={}},
    bidirected/.style={Latex-Latex,dashed},
    el/.style = {inner sep=2pt, align=left, sloped}
}

\graphicspath{{pic/}}
\newtheorem{theorem}{Theorem}

\newtheorem{lemma}{Lemma}

\theoremstyle{definition}

\newtheorem{example}{Example}

\algdef{SE}[DOWHILE]{Do}{doWhile}{\algorithmicdo}[1]{\algorithmicwhile\ #1}

\allowdisplaybreaks
\begin{document}

\title{Tackling Interference Induced by Data Training Loops in A/B Tests: A Weighted
Training Approach}
\author{Nian Si\thanks{Email: niansi@chicagobooth.edu.} \\ Booth School of Business, University of Chicago }%\and Jose Blanchet \\Department of Management Science and Engineering, Stanford University}
\date{\today }
\maketitle

\begin{abstract}
The standard data-driven pipeline in contemporary recommendation systems involves a continuous cycle in which companies collect historical data, train subsequently improved machine learning models to predict user behavior, and provide
improved recommendations. The user's response, which depends on the recommendation produced in this cycle, will become future training data. However, these data training-recommendation cycles can introduce interference in A/B tests, where data generated by control and treatment algorithms, potentially with different
distributions, are aggregated together. To address these challenges, we introduce a
novel approach called weighted training. This approach entails training a
model to predict the probability of each data point appearing in either the
treatment or control data and subsequently applying weighted losses during
model training. We demonstrate that this approach achieves the least
variance among all estimators that do not cause shifts in the training
distributions. Through simulation studies, we demonstrate the lower bias and
variance of our approach compared to other methods. 
\end{abstract}

\section{Introduction}

Experimentation (A/B tests) has emerged as the
standard method for evaluating feature and algorithmic updates in online platforms; see comprehensive guidance in \citet{kohavi2020trustworthy}. Instances of the use of A/B tests abound and are wide-ranging, from testing new pricing strategies in e-commerce, evaluating bidding strategies in online advertising, and updating and fine-tuning ranking algorithms in video-sharing platforms, just to name a few.

In such online platforms, recommendation systems are also in place to
enhance user experience by displaying relevant products and engaging videos.
The standard pipeline in recommendation systems operates as follows (as
illustrated in Figure \ref{fig:scheme}):  %\vspace{-0.1in}
\begin{enumerate}
	\item Using historical data, the system trains various machine learning (ML)
	models to predict users' behaviors, such as their interest in recommended
	items and their willingness to purchase certain products.  %\vspace{-0.1in}
	
	\item When a user request is received, the system identifies relevant items
	and ranks them based on the training scores generated by the machine
	learning models.
	 %\vspace{-0.1in}
	\item Items are recommended  to users based on the ranking.
 %\vspace{-0.1in}
	\item Users interact with the recommended items and take actions, including
	leaving comments below videos and making specific purchases.
	 %\vspace{-0.1in}
	\item The system records these user actions and feeds them back into the
	ML models, facilitating continuous model training.
\end{enumerate}
 %\vspace{-0.2in}
\begin{figure}[ht]
	\centering
	\begin{tikzpicture}[
		recstate/.style={rectangle, draw, text width=5cm, minimum width=5cm, align=center},]
		\node[recstate] (1) {Train models that predict users' behaviors};
		\node[recstate] (2) [below= of 1 ] {Rank items based on training scores};
		\node[recstate] (3) [below= of 2 ] {Recommend to users based on the ranking};
		\node[recstate] (4) [below =of 3 ] {User interacts with the recommended items};
		\node[recstate] (5) [below =of 4 ] {Record users' behaviors};
		\path (1) edge (2);
		\path (2) edge  (3);
		\path (3) edge (4);
		\path (4) edge (5);
		\draw [dashed] (5.west) to[bend left=30] node{Data} (1.west);
		%\phantom{\draw [dashed] (5) edge[bend right=90] node{Data}(1);}
	\end{tikzpicture}
	\caption{A standard pipeline in recommendation system}
	\label{fig:scheme}
\end{figure}
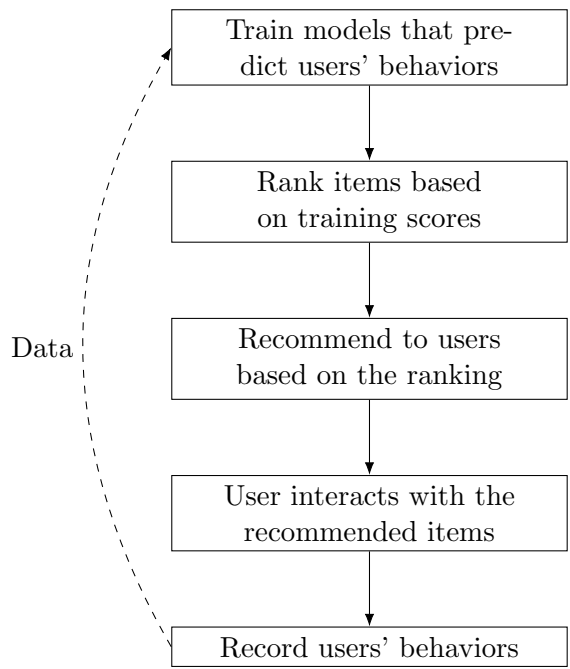

This pipeline ensures that the recommendation system continuously adjusts
and enhances its suggestions, taking into account user interactions and
feedback. However, it also generates a feedback loop, a phenomenon discussed
in both \citet{jadidinejad2020using} and \citet{chaney2018algorithmic}. As
we will demonstrate later, this feedback loop causes interference in A/B tests.

Interference, in the context of experimental design, means the violation of the Standard Unit Treatment Value Assumption (SUTVA) \citep{imbens2015causal}. According to SUTVA, the outcome for a given unit should solely depend on its treatment assignment and its own characteristics, and it should remain unaffected by the treatment assignments of other units. However, when data training loops are present, prior data generated under specific treatment assignments can lead to distinct model predictions. These predictions, in turn, can influence the outcomes observed for subsequent units, thereby violating the assumptions of SUTVA.

More specifically, let's consider a user-side experiment testing two distinct ranking
algorithms. In this scenario, we split the traffic in such a way that
control users are subjected to control algorithms, and treatment users are
subjected to treatment algorithms. Control and treatment algorithms generate data
that may follow different distributions. These data sets are then combined
and fed back into the ML models. This experimental procedure
is represented in Figure \ref{fig:scheme_ab}.

\begin{figure}[ht]
	\centering
	\begin{tikzpicture}
		\node[state] (1) {Training Models};
		\node[state] (2) [below left= of 1 ] {Control Algo};
		\node[state] (3) [below right= of 1] {Treatment Algo};
		\node[state] (4) [below =of 2 ] {Control Data};
		\node[state] (5) [below =of 3 ] {Treatment Data};
		\path (1) edge (2);
		\path (1) edge  (3);
		\path (2) edge (4);
		\path (3) edge (5);
		\draw [dashed] (4) edge[bend right=30] (1);
		\draw [dashed] (5) edge[bend left=30] (1);
	\end{tikzpicture}
	\caption{An A/B testing procedure}
	\label{fig:scheme_ab}
\end{figure}
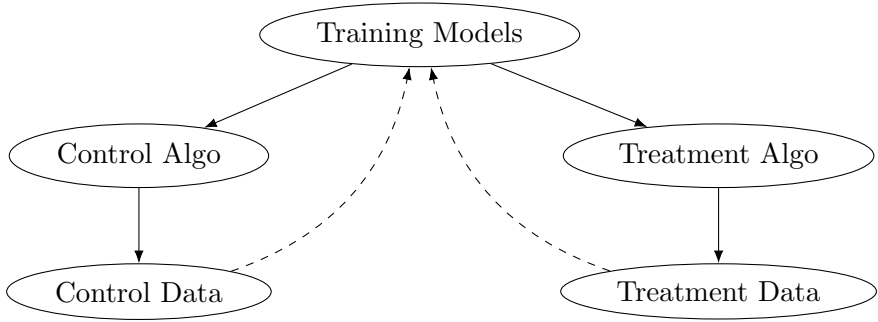
 %\vspace{-0.1in}
However, it's essential to recognize that this pooled distribution is
distinct from both the control data and the treatment data distributions. It
is widely acknowledged that variations in training distributions can lead to
significantly different predictions. 
To further illustrate this issue, let's consider the following example.

\begin{example}[Experimenting parameters of fusion formulas]
	Imagine a video-sharing platform with two distinct ML models
	that predict finishing rates (FR) and stay durations (SD), respectively. The
	platform's ranking algorithms rank videos using a linear fusion formula: 
		$\alpha_1 \text{FR} + \alpha_2\text{SD}.$
	In an A/B test, we aim to compare different parameter values $\{\alpha_1,
	\alpha_2\}$. Let us consider a scenario where the platform hosts two types of videos: short
	videos, which typically have high finishing rates and low stay durations,
	and long videos, which exhibit the opposite characteristics. If the
	treatment algorithm assigns a higher $\alpha_2$ to stay durations than the
	control algorithm, it will recommend more long videos in the treatment
	group. As a result, in the A/B tests, there will be a higher proportion of
	long videos in the pooled distribution. This can lead to different estimates
	of finishing rates and stay durations by the ML models,
	subsequently altering the recommendation outcomes produced by both the
	control and treatment algorithms.
\end{example}
 %\vspace{-0.1in}

This interference caused by data training loops closely relates to the
concept of ``symbiosis bias" recently introduced in \cite{holtz2023study}. In
their paper, they discuss cluster randomized designs and data-diverted
designs. Through simulations, they demonstrate that these designs can
effectively reduce biases compared to the naive approach.

In this paper, we introduce a weighted training approach. The concept
revolves around recognizing that a control data point may also appear in the
treatment data with a different probability. To harness this insight, we
create a new model that predicts the probability of each data point
appearing in either the treatment or control data. Subsequently, we train
the ML models using losses that are weighted based on these
predicted probabilities. By doing so, we demonstrate that if the weights are
accurately learned, there will be no shifts in the training distributions,
while making the most efficient use of available data.

The rest of the paper is organized as follows: Section \ref{sec:literature} discusses related literature on interference in A/B tests. Section \ref{sec:model} introduces a potential outcome framework modeling interference caused by data training loops. Section \ref{sec:algo} presents our weighted training approach along with theoretical justification. Section \ref{sec:numerical} showcases extensive simulation studies to demonstrate the performance of our proposed approach. Finally, we conclude with future works in Section \ref{sec:conclusion}.
 %\vspace{-0.1in}
\section{Related Literature}
\label{sec:literature}
 %\vspace{-0.1in}
\subsection{Interference in Experiments}
The existence of
interference is well-known in the literature. Empirical studies %
\citep{blake2014marketplace,holtz2020reducing,fradkin2015search} validate
that the bias caused by the interference could be as large as the treatment effect itself. 
In the following, we review the literature on various types of interference in
A/B tests.

\textbf{Interference in two-sided marketplaces.} In two-sided marketplaces, A/B tests are subject to
interference due to competition and spillover effects. %
\citet{johari2022experimental} and \citet{li2022interference} analyze biases
in both user-side and supply-side experiments using stylized models.
Additionally, \citet{bright2022reducing} consider a matching mechanism based
on linear programming and propose debiased estimators via shadow prices. To
mitigate bias, \citet{johari2022experimental} and \citet{bajari2021multiple}
introduce two-sided randomizations, which are also known as multiple
randomization designs. To measure the effectiveness of ``cold start" algorithms, \citet{ye2023cold} propose a similar yet different two-sided split design. Bipartite experiments are also introduced in %
\citet{eckles2017design,pouget2019variance,harshaw2023design}, where the
treatments are assigned in one group of units and the metrics are measured
in another group of units. Cluster experiments can also be applied in
marketplaces, as shown in \citet{holtz2020reducing,holtz2020limiting}.
Building on an equilibrium model, \citet{wager2021experimenting} propose a
local experimentation approach capable of accurately estimating small
changes in system parameters. Additionally, this idea has been extended by %
\citet{munro2021treatment}, who combined it with Bernoulli experiments to
estimate treatment effects of a binary intervention. For supply-side
(seller-side) experiments, \citet{ha2020counterfactual} and %
\citet{nandy2021b} put forth a counterfactual interleaving framework widely
implemented in the industry and \citet{Wang-ba-Producer2023} enhance the design with a novel tie-breaking rule to guarantee consistency and monotonicity. In the context of advertising experiments, %
\citet{liu2021trustworthy} propose a budget-split design and \citet{si2022optimal} use a weighted local linear regression estimation in situations where the budget  is not perfectly balanced between the treatment and control groups.

\textbf{Interference induced by feedback loops}. Feedback loops commonly exist in
complex systems. For instance, in the context of our earlier discussion in
the Introduction, data obtained from recommendations is fed back into the
underlying machine learning models. In online advertising platforms, the ads
shown previously can impact the subsequent ads' recommendations and bidding
prices, primarily due to budget constraints. However, there is relatively
limited literature that delves into experimental design dealing with
interference caused by feedback loops. 
\citet{goli2023bias} attempt to address such
interference by offering a bias-correction approach that utilizes data from past A/B
tests. In the context of searching ranking system, In the context of search
ranking systems, \citet{musgrave2023measuring} suggest the use of
query-randomized experiments to mitigate feature spillover effects. Additionally, for testing bandit learning algorithms, \citet{guo2023evaluating} propose a two-stage experimental design to estimate the lower bound and upper bound of the treatment effects.  
%Recently, \citet{zhu2023feedback} specifically study the challenges of
%the counterfactual interleaving design %
%\citep{ha2020counterfactual,nandy2021b} under interference induced by feedback
%loops. 
Furthermore, as mentioned earlier, \citet{holtz2023study} explore
similar issues to ours, which they refer to as ``Symbiosis Bias." 

\textbf{Markovian interference.} When a treatment can influence underlying
states, subsequently affecting outcomes in the following periods, we refer
to these experiments as being biased by Markovian interference. A classic
example is experimentation with different matching or pricing algorithms in
ride-sharing platforms. \citet{farias2022markovian} proposes a
difference-in-Q estimator for simple Bernoulli experiments, and its
performance is further validated through a simulation study with Douyin %
\citep{farias2023correcting}. Moreover, leveraging Markov decision
processes, optimal switchback designs have been analyzed in depth by %
\citet{glynn2020adaptive} and \citet{hu2022switchback}. In the specific
context of queuing, \citet{li2023experimenting} have conducted a study on
switchback experiments and local perturbation experiments. They have
discovered that achieving higher efficiency is possible by carefully
selecting estimators based on the structural information of the model.

\textbf{Temporal interference.} Temporal interference arises when there are
carry-over effects. Extensive investigations have been conducted on
switchback experiments %
\citep{bojinov2023design,hu2022switchback,xiong2023data,xiong2023bias}.
Besides switchback experiments, other designs %
\citep{basse2023minimax,xiong2019optimal} have also been proposed and proven
to be optimal in various contexts. In cases involving both spatial and
temporal interference, the new designs proposed in \citet{ni2023design}
combine both switchback experiments and clustering experiments.

\textbf{Network interference}. Network interference is frequently observed
in social networks, where a treatment unit's actions may have spillover
effects on their friends or neighbors. A substantial body of research has
looked into experimental design and causal inference under network
interference, with notable contributions from scholars such as %
\citet{hudgens2008toward}, \citet{gui2015network}, and \citet{li2022random},
among others. Specifically, the designs of graph cluster experiments, which
involves partitioning the graph into nearly disjointed clusters, has been
extensively investigated. This research area has seen contributions from
researchers 
\citep{aronow2017estimating,
	candogan2023correlated,ugander2013graph,ugander2023randomized} For a
comprehensive review of various approaches to address network interference,
we refer readers to Section 3 of \citet{yu2022estimating}.

In addition to the papers mentioned above, interference has also been
studied in other specialized settings. For instance, %
\citet{chawla2016b,basse2016randomization,liao2023statistical},  and \citet{liao2024interference} focus on
experimental design in auctions. \citet{han2023detecting} employ roll-outs,
a technique commonly implemented in experiments, to detect interference.
Additionally, \citet{boyarsky2023modeling} demonstrate that roll-outs can
also help estimation under stronger assumptions.

Finally, we briefly mention adaptive experimental design,  an emerging field aiming for more efficient experiments; see, for example, \citet{qin2022adaptivity, simchi2022multi,kuang2023weak,cook2023semiparametric}. 
 %\vspace{-0.2in}
 
\subsection{Feedback Loops in Recommendation Systems}
As modern platforms increasingly employ complex  systems, issues arising from feedback loops are becoming more pronounced. Researchers such as \citet{chaney2018algorithmic}, \citet{mansoury2020feedback}, and \citet{krauth2022breaking} have investigated problems related to the amplification of homogeneity and popularity biases due to feedback loops. Additionally, \citet{yang2023rectifying} and \citet{khenissi2022modeling} have noted that these feedback loops can lead to fairness concerns. The concept of user feedback loops and methods for debiasing them are discussed in \citet{pan2021correcting}, while \citet{jadidinejad2020using} consider how  feedback loops affect underlying models.
 In our work, we specifically focus on data training feedback loops and propose valid methods to address their impact on A/B tests.
\section{A Framework of A/B Tests Interfered by Data training Loops}

In this section, we construct a potential outcomes model %
\citep{imbens2015causal} for A/B tests that incorporate the training
procedures. Through our model, we will demonstrate the presence of
interference induced by data training loops in A/B tests.

We are focusing on user-side experiments, where users are assigned randomly
to the treatment group with a probability of $p$ and to the control group
with a probability of $1-p$.

Suppose there are $d$ features associated with each user-item pair, and the
system needs to predict $m$ different types of user behaviors (e.g.,
finishing rates, stay durations). We represent the feature space as $%
\mathcal{X}$, which is a subset of $\mathbb{R}^d$, and the outcome space as $%
\mathcal{Y}$, which is a subset of $\mathbb{R}^m$. In modern large-scale
recommendation systems, $d$ can be as extensive as billions, and $m$ can
encompass hundreds of different behaviors. We define a model class $\mathcal{%
	M}=\{M_{\theta },\theta \in \Theta \}$, which includes various models $%
M_{\theta }:\mathcal{X} \rightarrow \mathcal{Y}$. These models are
responsible for predicting user behaviors based on user-item features. In
this representation, we consolidate the prediction of $m$ distinct user
behaviors into a single model, which yields a $m$-dimensional output for the
sake of simplicity and convenience. In subsequent discussions, we will omit
the subscript $\theta $ for ease of notation.

At time $t$, the training model $M_{t}$ is trained from the previous model $%
M_{t-1}$ with additional data from time $t-1$, denoted as $\mathcal{D}_{t-1}$%
. This training process can be written as 
\begin{equation*}
	M_{t}=\digamma \left( M_{t-1},\mathcal{D}_{t-1}\right) ,
\end{equation*}%
where $\digamma $ denotes a training algorithm, e.g. stochastic gradient
descent (SGD) or Adam \citep{kingma2014adam}.

Further, at time $t$, we suppose there are $n_{t}$ new users have arrived.
For the $i$-user, $i=1,2,\ldots ,n_{t}$, the system recommends an item with
a feature vector $X_{i,t}=X_{i,t}\left( M_{t},Z_{i,t}\right) \in \mathbb{R}%
^{d},$ where $Z_{i,t}\in \left\{ 0,1\right\} $ denotes the treatment
assignment. Subsequently, the potential outcome for this user 
is given as $Y_{i,t}=Y_{i,t}\left( X_{i,t}\right) \in \mathcal{Y}$, which
represents the user's behaviors. Note that $Y_{i,t}$ is independent to $%
Z_{i,t}$ and $M_{t}$, given the feature vector $X_{i,t}.$ 
This assumption is grounded in the typical behavior of recommendation systems, where the primary influence on users' behaviors stems from the modification of recommended items.  Thus, $Y_{i,t}$ is not directly dependent on the treatment assignment $Z_{i,t}$ or the model state $M_{t}$ once the features $X_{i,t}$ are accounted for. We remark that our approach can be readily extended to cases where the treatment variable $Z$ directly affects the outcome $Y$,  as we shall see in Lemma \ref{lma:equal}. Due to the data
training loops, the data collected at time $t$ is incorporated into the
training dataset as follows: 
\begin{equation*}
	\mathcal{D}_{t}=\left\{ \left( X_{1,t},Y_{1,t}\right) ,\left(
	X_{2,t},Y_{2,t}\right) ,\ldots ,\left( X_{n_{t},t},Y_{n_t,t}\right) \right\} .
\end{equation*}
We plot the causal graph \citep{pearl2000models} in Figure \ref%
{figure:causal} to illustrate the dependence in the data training loops. 
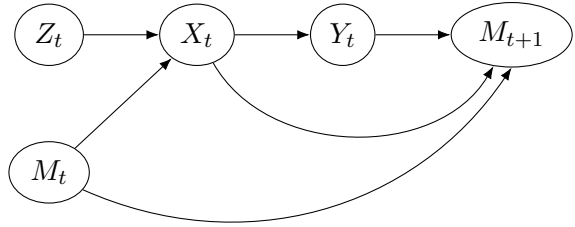
\begin{figure}[ht]
	\centering
	\begin{tikzpicture}
		\node[state] (1) {$Z_t$};
		\node[state] (2) [right =of 1] {$X_t$};
		\node[state] (3) [right =of 2] {$Y_t$};
		\node[state] (4) [right =of 3] {$M_{t+1}$};
		\node[state] (5) [below =of 1] {$M_t$};
		\path (1) edge   (2);
		\path (2) edge   (3);
		\path (3) edge   (4);
		\path (5) edge   (2);
		\path 	(2)   edge[bend right=60]    (4);
		\path (5) edge[bend right=40]   (4.south);
	\end{tikzpicture}
 %\vspace{-0.1in}
	\caption{Dependence of different objects in the data training loops, where
		we omit the subscript $i$ for simplicity }
	\label{figure:causal}
\end{figure}
 %\vspace{-0.1in}

It's important to note that $\mathcal{D}_{t}$ consists of
recommendation data, which may differ from the control and treatment data.
Consequently, when applying the training algorithm $\digamma$, the model at
the next time step, $M_{t+1}$, will differ from the model trained solely on
control or treatment data. This, in turn, impacts the recommendations $%
X_{\cdot ,t+1}$ at the subsequent period. Therefore, it becomes evident that
these A/B tests are susceptible to interference caused by data training
loops.

Our objective is to estimate the global treatment effect (GTE), which is
defined as the difference between the metrics observed under the global
treatment and the global control regimes. The global treatment regime is
defined as having all $Z_{i,t}$ equal to one, while the global control
regime is defined as having all $Z_{i,t}$ equal to zero. In mathematical
terms, we represent this as follows: within the global treatment regime, the procedure is outlined as:
\begin{eqnarray*}
	X_{i,t}^{\mathrm{GT}} &=&X_{i,t}\left( M_{t}^{\mathrm{GT}},1\right)
	,Y_{i,t}^{\mathrm{GT}}=Y_{i,t}\left( X_{i,t}^{\mathrm{GT}}\right) , \\
	\mathcal{D}_{t}^{\mathrm{GT}} &=&\left\{ \left( X_{1,t}^{\mathrm{GT}%
	},Y_{1,t}^{\mathrm{GT}}\right) ,\left( X_{2,t}^{\mathrm{GT}},Y_{2,t}^{%
		\mathrm{GT}}\right) ,\ldots ,\left( X_{n_{t},t}^{\mathrm{GT}},Y_{n_t,t}^{%
		\mathrm{GT}}\right) \right\} \\
	M_{t}^{\mathrm{GT}} &=&\digamma \left( M_{t-1}^{\mathrm{GT}},\mathcal{D}%
	_{t-1}^{\mathrm{GT}}\right) ,\text{ \ for }t=1,\ldots ,T; 
\end{eqnarray*}
Similarly,  within the global control regime, we have:
\begin{eqnarray*}
	X_{i,t}^{\mathrm{GC}} &=&X_{i,t}\left( M_{t}^{\mathrm{GC}},0\right)
	,Y_{i,t}^{\mathrm{GC}}=Y_{i,t}\left( X_{i,t}^{\mathrm{GC}}\right) , \\
	\mathcal{D}_{t}^{\mathrm{GC}} &=&\left\{ \left( X_{1,t}^{\mathrm{GC}%
	},Y_{1,t}^{\mathrm{GC}}\right) ,\left( X_{2,t}^{\mathrm{GC}},Y_{2,t}^{%
		\mathrm{GC}}\right) ,\ldots ,\left( X_{n_{t},t}^{\mathrm{GC}},Y_{n_t,t}^{%
		\mathrm{GC}}\right) \right\} , \\
	M_{t}^{\mathrm{GC}} &=&\digamma \left( M_{t-1}^{\mathrm{GC}},\mathcal{D}%
	_{t-1}^{\mathrm{GC}}\right) ,\text{ \ for }t=1,\ldots ,T, 
\end{eqnarray*}%
Here, we assume $\mathcal{D}_{0}^{\mathrm{GC}} =\mathcal{D}_{0}^{\mathrm{GT%
}}\text{ and }M_{0}^{\mathrm{GC}}=M_{0}^{\mathrm{GT}}$. 
The $m$-dimensional GTE is defined as 
\begin{equation*}
	\mathrm{GTE}=\mathbb{E}\left[ \frac{1}{\sum_{t=1}^{T}n_{t}}%
	\sum_{t=1}^{T}\sum_{i=1}^{n_{t}}\left( Y_{i,t}^{\mathrm{GT}}-Y_{i,t}^{%
		\mathrm{GC}}\right) \right] .
\end{equation*}

In the naive A/B tests, the estimator is 
\begin{align}
	\frac{1}{\sharp \left\{ Z_{i,t}=1\right\} }\sum_{Z_{i,t}=1}^{{}}Y_{i,t}%
	\left( X_{i,t}\left( M_{t},1\right) \right) - 
	 \frac{1}{\sharp \left\{
		Z_{i,t}=0\right\} }\sum_{Z_{i,t}=0}^{{}}Y_{i,t}\left( X_{i,t}\left(
	M_{t},0\right) \right) ,  \label{naive_estimator}
\end{align}%
where $\sharp \left\{ Z_{i,t}=1\right\} $ and $\sharp \left\{
Z_{i,t}=0\right\} $ are the number of users in the treatment and control,
respectively. Because of the interference induced by data training loops, it
is possible for the estimator to exhibit bias when estimating the Global
Treatment Effect (GTE). \label{sec:model}

\section{A Weighted Training Approach}

Based on the potential outcome model established in Section \ref{sec:model},
it becomes apparent that interference arises due to shifts in the training
distributions. In this section, we will introduce an approach that assigns
weights to the original data distributions obtained from the A/B tests. We
will demonstrate that these weighted distributions have the capability to
recover the data distributions for the control group and the treatment group.

In abstract terms, constructed in a probability space $\left( \Omega ,%
\mathcal{F},P\right) ,$ let $D=(X,Y)$ be the random variable representing
some data of $(X,Y)\in \mathcal{X}\times \mathcal{Y}$. Specifically, $%
D_{C}=(X_{C},Y_{C}),D_{T}=(X_{T},Y_{T})$ be the random variable representing
control data and treatment data, respectively. We use $\mathcal{D}_{C},%
\mathcal{D}_{T}$ to denote the distributions of the control data and
treatment data, respectively. Therefore, by using $\mathcal{L}(\cdot )$ to
denote the law (distribution) of a random variable, we have 
\begin{equation*}
	\mathcal{D=\mathcal{L}}\left( D\right) ,\mathcal{D}_{C}=\mathcal{L}\left(
	D_{C}\right) \text{ and }\mathcal{D}_{T}=\mathcal{L}\left( D_{T}\right) .
\end{equation*}%
Let the treatment assignment $Z$ also be constructed in the same probability
space. Importantly, $Z$ is independent to $\left\{ D_{C},D_{T}\right\} ,$
i.e., 
\begin{equation*}
	Z\bot \left\{ D_{C},D_{T}\right\} ,
\end{equation*}%
which is the unconfoundedness assumption in casual inference %
\citep{rosenbaum1983central}. The random variable $D_{E}=\left\{
X_{E},Y_{E}\right\} $ represents the data obtained from the experiment and
can be expressed as follows: 
\begin{equation*}
	D_{E}=D_{T}Z+D_{C}\left( 1-Z\right) ,
\end{equation*}%
where $P(Z=1)=p$ represents the probability of treatment assignment.
Consequently, the distribution of the experimental data can be described as: 
\begin{equation*}
	\mathcal{D}_{E}=p\mathcal{D}_{T}+(1-p)\mathcal{D}_{C},
\end{equation*}%
due to the independence of $Z$ and $\left\{ D_{C},D_{T}\right\} .$

%To emphasize the model's dependence on the training distribution $\mathcal{D}$, we represent it as $M(\mathcal{D)}$. 
Our objective is to shift the
distribution of experimental data $\mathcal{D}_{E}$ towards that of the
control data $\mathcal{D}_{C}$ and the treatment data $\mathcal{D}_{T}$ to
mitigate bias. To achieve this, we introduce a weighting function $W(\cdot):%
%
%
%
%
%
%
%\mathcal{X}\times \mathcal{Y}\times\{0,1\}
\Omega \rightarrow \mathbb{R}_{+}$, with the property that $\mathbb{E}[W]=1$%
. We denote the resulting weighted distribution as $W\mathcal{D}$, i.e.,
\begin{equation*}	W\mathcal{D(}A\mathcal{)=}\mathbb{E}\left[ WI\left\{ D\in A\right\} \right] 
	\text{ for any measurable }A\text{ in }\mathcal{X}\times \mathcal{Y},
\end{equation*}%\
where we primarily focus on $\mathcal{D}= \mathcal{D}_E$ and $D=D_E$ in this paper. 
It is easy to check $W\mathcal{D}$ in $\mathcal{X}\times \mathcal{Y}$ is
also a probability distribution as $W(\cdot )$ is non-negative and $\mathbb{E%
}[W]=1.$

Our first result, presented below, demonstrates that by selecting the weight
function as $\mathbb{E}\left[ Z|X_{E}\right] /p$ or $\left( 1-\mathbb{E}%
\left[ Z|X_{E}\right] \right) /(1-p)$, we can effectively recover the
treatment and control data distributions, respectively.

\begin{lemma}
	The weighted functions 
	\begin{align*}
		W_{T}(X_{E},Y_{E},Z)=\frac{\mathbb{E}\left[ Z|X_{E}\right] }{p}\text{ and } W_{C}(X_{E},Y_{E},Z)=\frac{1-\mathbb{E}\left[ Z|X_{E}\right] }{1-p}
		\label{weight_function}
	\end{align*}%
	satisfy 
	\begin{equation*}
		W_{T}\mathcal{D}_{E}\overset{d}{\mathcal{=}}\mathcal{D}_{T}\text{ and }W_{C}%
		\mathcal{D}_{E}\overset{d}{\mathcal{=}}\mathcal{D}_{C},
	\end{equation*}%
	where $\overset{d}{\mathcal{=}}$ means equal in distribution.
%	 and here $%
%	W_{C}(X_{E},Z)$ and $W_{T}(X_{E},Z)$ should be understood as $%
%	W_{C}(X_{E}(\omega ),Z(\omega ))$ and $W_{T}(X_{E}(\omega
%	),Y_{E}(\omega ),Z(\omega )),$ for any $\omega \in \Omega $.
 \label	{lma:equal}
\end{lemma}
\textbf{Remark:} In cases where the treatment variable $Z$ is able to  directly affect the outcome $Y_E$, the adjustment can be made by substituting the conditional expectation $\mathbb{E}\left[ Z|X_{E}\right]$ with $\mathbb{E}\left[ Z|X_{E},Y_{E}\right]$.

The proof of Lemma \ref{lma:equal} is presented in Appendix \ref%
{appendix:proof}. Lemma \ref{lma:equal} shows that we are able to
reconstruct the treatment and control data distributions from the A/B
testing data distribution, provided that we can estimate $\mathbb{E}\left[
Z|X_{E}\right] $ with sufficient accuracy. 

Since the quantity $\mathbb{E}\left[Z|X_{E}\right]$ is typically unknown beforehand, it becomes necessary to estimate it from the available data. To achieve this, we construct an additional machine learning model denoted as ${G}_{\theta_W}$. This model is trained using the data $\{X_E,Z\}$ obtained from the experiments, treating it as a classification problem. Subsequently, the predictions generated by ${G}_{\theta_W}$ are utilized as weights (after proper normalization) to form weighted losses for the original machine learning models. This method is detailed in Algorithm \ref{algo:weighted}.

\begin{algorithm}[!ht]
	\caption{A weighted training approach for A/B tests}
	\label{algo:weighted}
	\begin{algorithmic}[1]
		\Require{The probability of treatment assignment: $p$; a model class for the weight prediction: $\mathcal{G}=\{G_{\theta_W}: \mathbb{R}^d \rightarrow [0,1], {\theta_W}\in \Theta_W\}$; the machine learning model class: $\mathcal{M}=\{M_\theta: \mathcal{X}  \rightarrow \mathcal{Y} , \theta\in \Theta\}$; loss functions: $\ell(M(X),Y)$ (could be $m$-dimensional).}
	    \State Initialize two models, the treatment model ${M}_{\theta_T}$ and the control model ${M}_{\theta_C}$, both of which are set to the current production model.
		\For{$t \gets 1$  to the end of the experiment  }
		\For{$i \gets 1$  to $n_t$  } 
		\State User $i$ arrives. The platform randomly assigns user $i$ to the treatment group with probability $p$. 
		\State When a user is assigned to the treatment group, the platform recommends an item based on the treatment algorithm and model, and vice versa.
		\State Collect data $(X_{i,t},Y_{i,t},Z_{i,t})$.
		\EndFor 
		\State Compute weights: 
		$$W_{T,i,t}=\frac{G_{\theta_W}(X_{i,t})}{p} \text{ and } W_{C,i,t}=\frac{1-G_{\theta_W}(X_{i,t})}{1-p} \text{ , for } i=1,2,\ldots,n_t. $$ 
		\State Update the treatment model  ${M}_{\theta_T}$ by minimizing the weighted loss 
	$$\frac{1}{n_t}\sum_{i=1}^{n_t}W_{T,i,t}\ell({M}_{\theta_T}(X_{i,t}),Y_{i,t}).$$
	    \State Update the control model  ${M}_{\theta_C}$ by minimizing the weighted loss 
	    $$\frac{1}{n_t}\sum_{i=1}^{n_t}W_{C,i,t}\ell({M}_{\theta_C}(X_{i,t}),Y_{i,t}).$$
		\State Update the model $G_{\theta_W}$ using data $\{(X_{i,t},Z_{i,t}),i=1,\ldots,n_t\}$.
		\EndFor
		
		\Return the estimator (\ref{naive_estimator}).
	\end{algorithmic}
\end{algorithm}

We remark that while $\mathbb{E}\left[Z|X_{E}\right]$ might be complex, there is no need for precise estimation in practical applications. In fact, simple models like two-layer neural networks perform well, as demonstrated in our numerical results.

From the proof of Lemma \ref{lma:equal}, one may note that the simple weight
function $\tilde{W}=Z$ also satisfy 
$
	\tilde{W}\mathcal{D}_{E}\overset{d}{\mathcal{=}}\mathcal{D}_{T}.
$
Indeed, using $Z_{i,t}$ instead of training a model $G_{\theta_W}$ in Algorithm \ref{algo:weighted}
results in a data splitting approach, also known as a data-diverted
experiment, as discussed in \citet{holtz2023study}. In such experiments,
each model is updated exclusively using data generated by users exposed to
the corresponding algorithm. However, this approach lacks data efficiency,
as it utilizes only a fraction of the data, namely $p$ for the treatment
model and $1-p$ for the control model. 
For instance, in cases where the control data distribution is identical to
the treatment data distribution, our approach can leverage all available
data for training both control and treatment models. This is because $\frac{%
	\mathbb{E}\left[ Z|X_{E}\right] }{p}=\frac{1-\mathbb{E}\left[ Z|X_{E}\right] 
}{1-p}=1$ in this case.

Intuitively, in the finite sample regime with $n$ samples, the variance of
the estimator should be proportional to $\frac{1}{n^{2}}\sum_{i=1}^{n}\left(
W_{i}/p\right) ^{2}.$ In the following, we will demonstrate that our
approach can achieve this lower variance, defined in this manner, among all
possible weights without causing shifts in the training distributions.

\begin{theorem}
	\label{thm:min_var} $W_{T}(X_{E},Y_{E},Z)=\mathbb{%
		E}\left[ Z|X_{E}\right] /p$ attains the minimum of the following
	optimization problem 
	\begin{equation}
		\min_{W(\cdot ):\Omega \rightarrow \mathbb{R}_{+}}\left\{ \mathbb{E}\left[
		W^{2}\right] :W\mathcal{D}_{E}\overset{d}{\mathcal{=}}\mathcal{D}%
		_{T}\right\} .  \label{problem:treatment_var}
	\end{equation}%
	Similarly, $W_{C}(X_{E},Y_{E},Z)=\left( 1-\mathbb{%
		E}\left[ Z|X_{E}\right] \right) /\left( 1-p\right) $ attains the minimum of
	the following optimization problem%
	\begin{equation}
		\min_{W(\cdot ):\Omega \rightarrow \mathbb{R}_{+}}\left\{ \mathbb{E}\left[
		W^{2}\right] :W\mathcal{D}_{E}\overset{d}{\mathcal{=}}\mathcal{D}%
		_{C}\right\} .  \label{problem:control_var}
	\end{equation}
\end{theorem}

Theorem 1 implies that our proposed weights, $\frac{\mathbb{E}\left[ Z|X_{E}%
	\right]}{p}$ and $\frac{1-\mathbb{E}\left[ Z|X_{E}\right]}{1-p}$, achieve
maximum data efficiency while adhering to the constraint of no training
distributional shifts. The proof of Theorem \ref{thm:min_var} is provided in
Appendix \ref{appendix:proof}. \label{sec:algo}

\section{Numerical Results}
\label{sec:numerical}
In this section, we present simulation results. In subsection \ref%
{subsection:simu}, we specify the simulation setup and the implementation
details. In subsection \ref{subsection:AB}, we simulate A/B tests to
demonstrate the lower bias and variance of our approach compared to other
methods. In subsection \ref{subsection:AA}, we simulate A/A tests to compare
type I errors of different methods. Additional experiments and results can be found in Appendix \ref{appendix:additional}.

\subsection{Simulation Setups}

\label{subsection:simu} We conducted a simulation inspired by Example 1 in
 Introduction. In this simulation, we consider two types of videos: long
and short, and the recommendation system relies on two metrics: finishing
rates (FR) and stay durations (SD). Users arrive sequentially, and for each
user, there are a total of $N=100$ candidate videos available. These videos
are divided into two equal groups, with half of them being long videos and
the other half being short videos. The platform selects one video from this
pool to show to each user. Furthermore, we assume that the features for
user-video pairs are 10-dimensional, following independent uniform
distributions in the range [0,1]. Additionally, we assume linear models that%
\begin{eqnarray*}
\mathrm{FR}_{\mathrm{short}} &=&\text{Sigmoid}(\beta _{\mathrm{FR,short}}^{\top }X-2.5), \\
	\mathrm{FR}_{\mathrm{long}} &=&\text{Sigmoid}(\beta _{\mathrm{FR,long}}^{\top }X-2.5), \\
\mathrm{SD}_{\mathrm{short}} &\sim &\exp \left( \beta _{\mathrm{SD,short}%
}^{\top }X\right) , \\
\mathrm{SD}_{\mathrm{long}} &\sim &\exp \left( \beta _{\mathrm{SD,long}%
}^{\top }X\right),
\end{eqnarray*}%
where $\mathrm{Sigmoid}$ means the sigmoid function and $\exp \left( \cdot \right) $ means an exponential distribution and 
\begin{eqnarray*}
\beta _{\mathrm{FR,short}} &=&0.9\times \lbrack 0,0.1,0.2,\ldots ,0.9], \\
\beta _{\mathrm{FR,long}} &=&0.6\times \lbrack 0,0.1,0.2,\ldots ,0.9], \\
\beta _{\mathrm{SD,short}} &=&[1,0.9,0.8\ldots .0.1], \\
\beta _{\mathrm{SD,long}} &=&1.5\times \lbrack 1,0.9,0.8\ldots .0.1].
\end{eqnarray*}%
The user's decision to finish watching a video or not follows a Bernoulli
distribution with a probability equal to the finishing rate. By
setting the parameters in this manner, we ensure that short videos generally
have high finishing rates and short stay durations, while long videos
exhibit the opposite characteristics.

The machine learning models employ logistic regression for predicting
finishing rates and linear regression for predicting stay durations. The feature set consists of 10 user-video pair features, along with an
indicator variable that specifies whether the video is long or short. It's
important to note that there is a model misspecification present, as the
true parameters for long and short videos are different. In our machine
learning models, we assume these parameters to be equal, but we introduce an
additional parameter corresponding to the video length indicator for an
adjustment.

We employ Stochastic Gradient Descent (SGD) to train both machine learning
models, employing a batch size of $B=n_1=n_2=\ldots=n_T=128$ for all time
steps. The learning rate is set to $0.1$. Throughout all simulations, we
maintain a fixed value of $T=10000$. Consequently, the total number of users
involved in the experiments amounts to 1,280,000.

The platform recommends the video that yields the highest value among the
100 candidate videos based on the following formula: 
\begin{equation*}
\alpha \widehat{\text{FR}} + \widehat{\text{SD}},
\end{equation*}
where $\widehat{\text{FR}}$ and $\widehat{\text{SD}}$ represent the
predictions generated by the machine learning models. The A/B tests are
designed to assess the difference between two distinct $\alpha$ values. We
focus on three metrics, FR, SD, and the proportion of short videos on the
platform.

We compare our approach to three other methods: data pooling, snapshot, and
data splitting methods.

\begin{itemize}
\item \textbf{Data pooling:} This is the standard naive approach, where
machine learning models are trained on the combined control and treatment
data.

\item \textbf{Snapshot:} In this method, the machine learning models are
never retrained during the A/B tests. Predictions are solely based on the
models' initial snapshot at the beginning of the experiments.

\item \textbf{Data splitting:} Also known as data-diverted, as discussed in %
\citet{holtz2023study}, each model is exclusively trained on the data
obtained from its respective algorithm.
\end{itemize}

While \citet{holtz2023study} also explore cluster randomized experiments, it's worth noting that in our specific context, determining how to cluster users presents challenges. Consequently, we do not make direct comparisons with cluster randomized experiments.   
As we discussed in Section \ref{sec:algo}, the data splitting method may
encounter several challenges:

\begin{itemize}
\item \textbf{High variance.} Since machine learning models can only see a
portion of the data, the lack of data efficiency may lead to high variance
in model estimators, resulting in increased variance in the experimental
metrics.

\item \textbf{External validity.} In our simulation, the data splitting
method is equivalent to reducing the batch size. It is well-known that batch
size plays a crucial role in machine learning, and different batch sizes can
yield fundamentally different performances. Therefore, treatment effect
estimates in scenarios with small batch sizes may not accurately predict
treatment effects in scenarios with large batch sizes, compromising external
validity.

\item \textbf{Experimentation costs.} In today's online platforms, thousands
of experiments run each day. Consequently, experimentation costs cannot be
overlooked, even though each experiment only runs for a relatively short
period. Reducing the data size can compromise the performance of the machine
learning model, potentially leading to suboptimal recommendations and
increased experimentation costs.
\end{itemize}

In our approach, we employ a two hidden layer fully connected network with ReLU
activations to train the weighting model $G_{\theta_W}$. Each layer comprises 64 neurons,
and we utilize the Adam optimizer \citep{kingma2014adam} with a learning
rate of 0.001. Our training process for the weighting model commences after the initial 200 periods. During these initial 200 periods, the control and treatment machine learning models are trained as in the data splitting method.

Subsequently, we will conduct the A/B tests 100 times and create violin
plots to visualize the estimated treatment effects.

\subsection{A/B Tests}

\label{subsection:AB} We first examine the comparison between the control
parameter $\alpha_C=10$ and the treatment parameter $\alpha_T=9$ with a
treatment assignment probability of $p = 1/2$. The logloss of our weighting model is about 0.96 (base 2). This value indicates a performance slightly above that of a purely random guess, which would yield a logloss of 1. Despite the marginal improvement over randomness in terms of logloss, it's crucial to highlight that this translates into significant gains in the accuracy of treatment effect estimation, as we shall see in Figure \ref{fig:ab_p=1/2} and Table \ref{tab:bias_AB_1/2}.
\begin{figure}[!ht]
	\centering 
	\subfigure[treatment effects: proportion]{
		\includegraphics[width=2.1in]{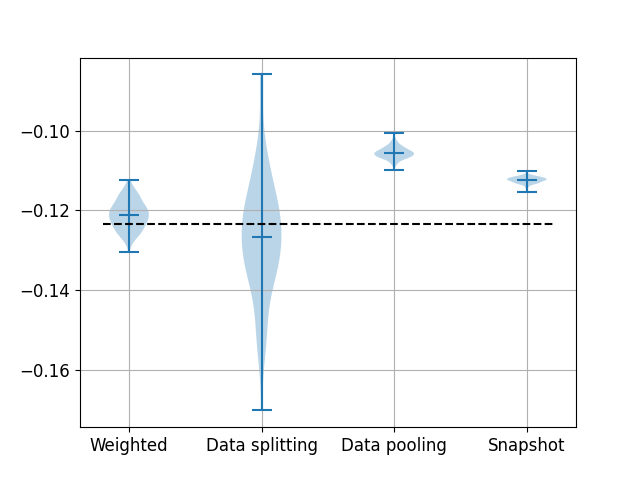}}
	\subfigure[treatment effects: SD]{
		\includegraphics[width=2.1in]{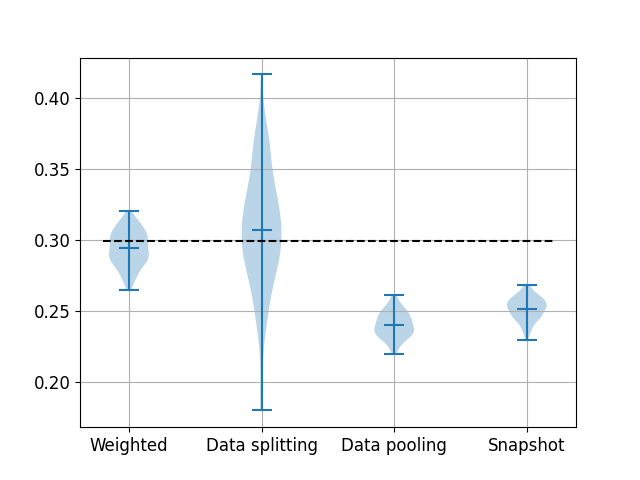}}
	\subfigure[treatment effects: FR]{
		\includegraphics[width=2.1in]{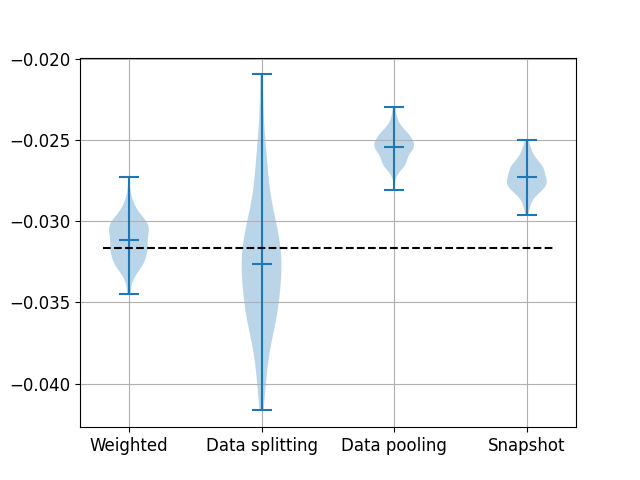}}
	\subfigure[treatment: proportion]{ 
		\includegraphics[width=2.1in]{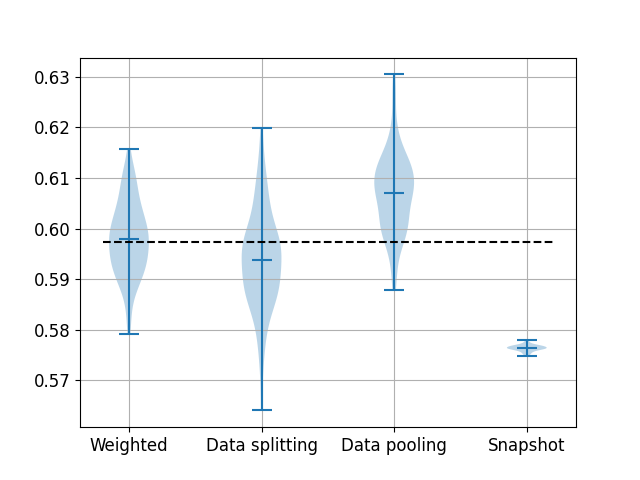}}
	\subfigure[treatment: SD]{ \includegraphics[width=2.1in]{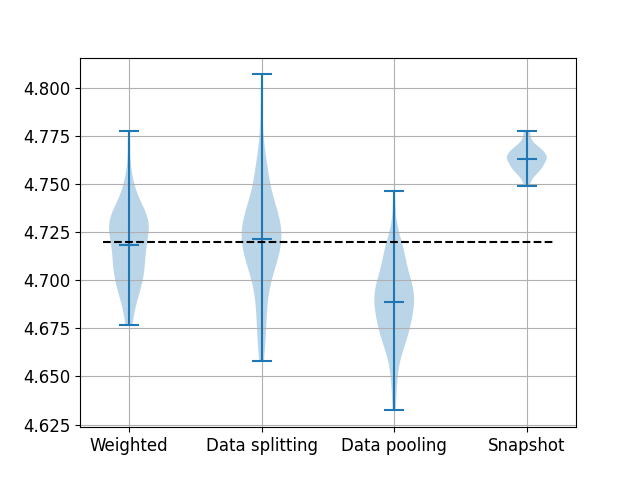}}
	\subfigure[treatment: FR]{
		\includegraphics[width=2.1in]{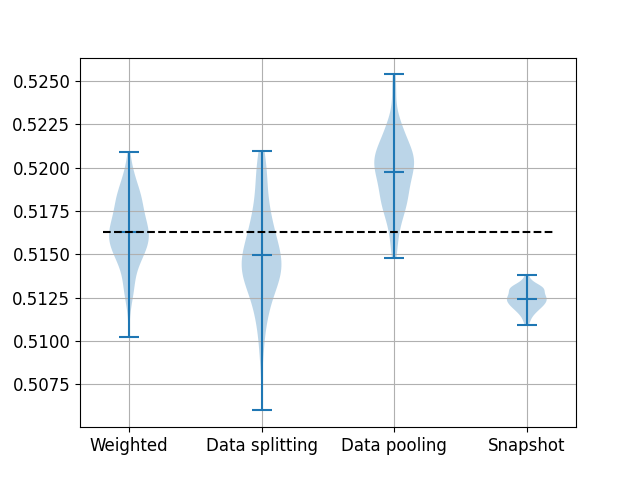}}
	\subfigure[control: proportion]{
		\includegraphics[width=2.1in]{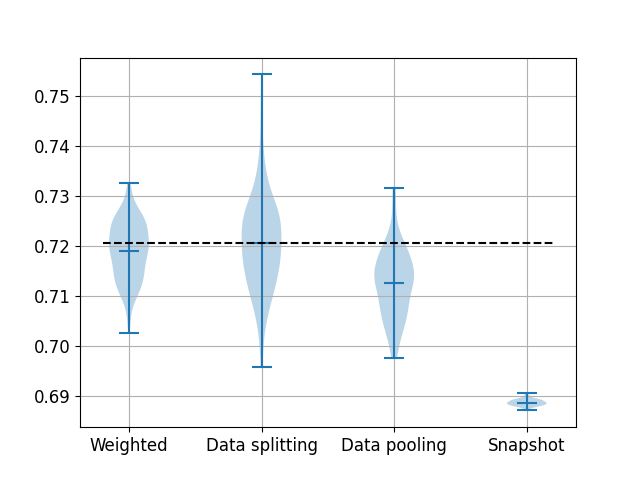}}
	\subfigure[control: SD]{
		\includegraphics[width=2.1in]{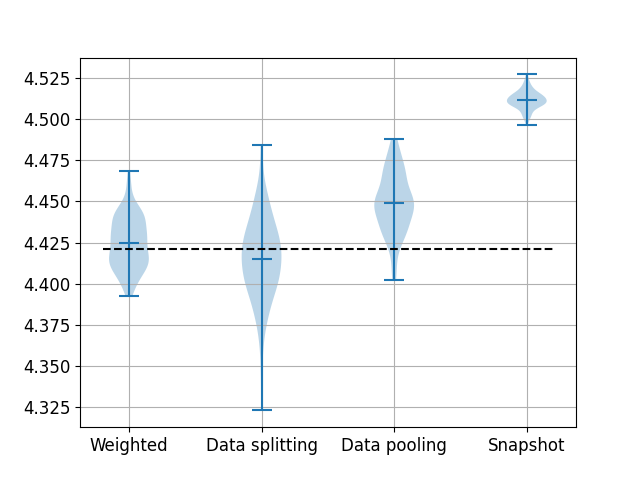}}
	\subfigure[control: FR]{
		\includegraphics[width=2.1in]{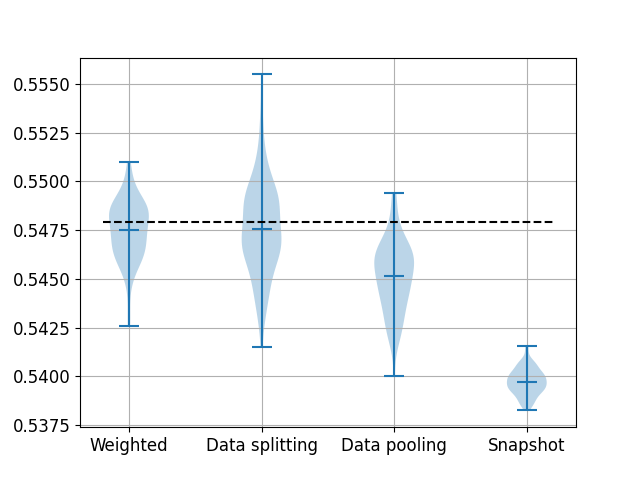}}
	\caption{A/B testing results for $\protect\alpha_C=10$, $\protect\alpha_T=9$%
		, and $p = 1/2$ }
	\label{fig:ab_p=1/2}
\end{figure}

In the figure, the black dotted line represents the true global treatment
effects (GTE), which have been computed through simulation. We present
various estimators obtained from 100 independent A/B tests along with their
respective mean, lower, and upper bounds. Specifically, we provide results
for treatment effects, global treatment, and global control regimes in the
first, second, and third rows, respectively. Additionally, we report results
for the proportion of short videos, SD, and FR in the first, second, and
third columns, respectively.

We further provide information on the bias and standard errors of treatment effect estimators obtained using various methods in Table \ref{tab:bias_AB_1/2}. In each metric, the first column represents the bias in comparison to the true global treatment effect (GTE). The second column displays the standard deviation calculated from the results of the 100 A/B tests. Lastly, the third column showcases the standard error estimates obtained through two-sample t-tests in a single A/B test
i.e.,%
\begin{equation*}
\mathrm{SE}=\sqrt{\frac{\mathrm{Var}(Y|Z=1)}{\sharp \left\{
Z_{i,t}=1\right\} }+\frac{\mathrm{Var}(Y|Z=0)}{\sharp \left\{
Z_{i,t}=0\right\} }}.
\end{equation*}
\begin{table}[!htb]
	\caption{Bias, standard deviation, and standard error estimated from the
		experiment for the metrics in the case that $\protect\alpha_C=10$, $\protect%
		\alpha_T=9$, and $p = 1/2$}
	\label{tab:bias_AB_1/2}\centering
	\begin{tabular}{lccccccccc}
		\toprule & \multicolumn{3}{c}{Proportion of short videos} & 
		\multicolumn{3}{c}{Stay durations} & \multicolumn{3}{c}{Finishing rates}  \\ 
		\midrule & Bias & STD & SE & Bias & STD & SE & Bias & STD & SE \\ 
		Weighted & 0.002 & 0.004 & 0.001 & -0.005 & 0.012 & 0.008 & 0.000 & 0.001 & 
		0.001 \\ 
		Data splitting & -0.003 & 0.015 & 0.001 & 0.008 & 0.042 & 0.008 & -0.001 & 
		0.004 & 0.001 \\ 
		Data pooling & 0.018 & 0.002 & 0.001 & -0.059 & 0.009 & 0.008 & 0.006 & 0.001
		& 0.001 \\ 
		Snapshot & 0.011 & 0.001 & 0.001 & -0.047 & 0.008 & 0.009 & 0.004 & 0.001 & 
		0.001 \\ 
		\bottomrule 
	\end{tabular}%
\end{table}

From Figure \ref{fig:ab_p=1/2} and Table \ref{tab:bias_AB_1/2}, it is evident that our approach consistently demonstrates the lowest bias across all metrics compared to other approaches. The data splitting method also manages to achieve relatively low biases but exhibits significantly higher variance. Furthermore, it's worth noting that the true variance of the data splitting estimator is considerably larger than the standard error estimated from a two-sample t-test. Consequently, this could potentially lead to confidence intervals that underestimate the true level of variability.


% Table generated by Excel2LaTeX from sheet '10VS9 0.5'

In Table \ref{tab:cost_AB_1/2}, we have calculated the experimentation costs. For treatment users, we computed the average treatment values based on the treatment linear fusion formula, i.e.,
$$
 \frac{1}{N_T} \sum_{i=1}^{N_T} \alpha_T \text{Finish}_i + \text{StayDuration}_i, 
$$
where $N_T$ represents the number of users in the treatment group and $\text{Finish}_i$ and $\text{StayDuration}_i$ indicate whether a user finished watching a video and their duration of stay, respectively.  While for control users, we averaged the control values based on the control linear fusion formula:
$$
 \frac{1}{N_C} \sum_{i=1}^{N_C} \alpha_C \text{Finish}_i + \text{StayDuration}_i. 
$$
It's apparent that our approach is only slightly worse than the global treatment/control regime, and the data splitting method incurs the lowest costs, indicating that it results in higher experimental expenses.
\begin{table}[htbp]
\caption{Experimentation values in the case that $\protect\alpha_C=10$, $%
\protect\alpha_T=9$, and $p = 1/2$}
\label{tab:cost_AB_1/2}\centering
\begin{tabular}{lcc}
\toprule & Treatment values & Control values \\ 
\midrule Global & 9.8827 $\pm$ 0.0006 & 9.3523 $\pm$ 0.0006 \\ 
Weighted & 9.8816 $\pm$ 0.0009 & 9.3521 $\pm$ 0.0008 \\ 
Data splitting & 9.8710 $\pm$ 0.0008 & 9.3431 $\pm$ 0.0008 \\ 
Data pooling & 9.8861 $\pm$ 0.0008 & 9.3551 $\pm$ 0.0009 \\ 
Snapshot & 9.8876 $\pm$ 0.0009 & 9.3692 $\pm$ 0.0008 \\ 
\bottomrule &  & 
\end{tabular}%
\end{table}

We proceed to conduct simulations for $p = 0.2$, with the same treatment and control parameters, i.e., $\alpha_C=10$ and $\alpha_T=9$. This scenario is arguably more relevant, as A/B tests in practice often begin with smaller treatment proportions. The results are presented in Figure \ref{fig:ab_p=0.2}, and detailed bias, variance, and cost findings can be found in Tables \ref{tab:bias_AB_0.2} and \ref{tab:cost_AB_0.2}.

Once again, our approach demonstrates the lowest bias and reasonable variance. However, it's important to note that in this case with $p=0.2$, the data splitting method exhibits higher bias and variance compared to the simulation with $p=1/2$.
\begin{figure}[!ht]
\centering 
\subfigure[treatment effects: proportion]{
		\includegraphics[width=2.1in]{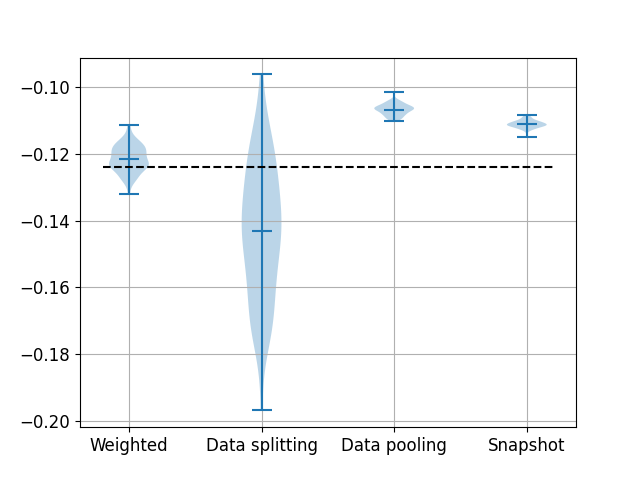}}
\subfigure[treatment effects: SD]{
		\includegraphics[width=2.1in]{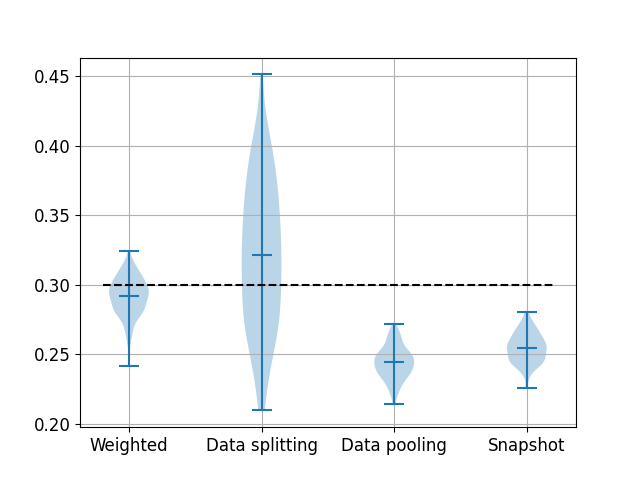}}
\subfigure[treatment effects: FR]{
\includegraphics[width=2.1in]{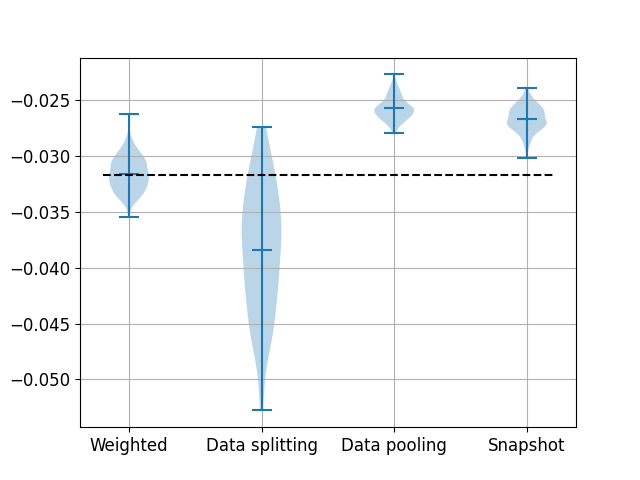}}
\subfigure[treatment: proportion]{ 
		\includegraphics[width=2.1in]{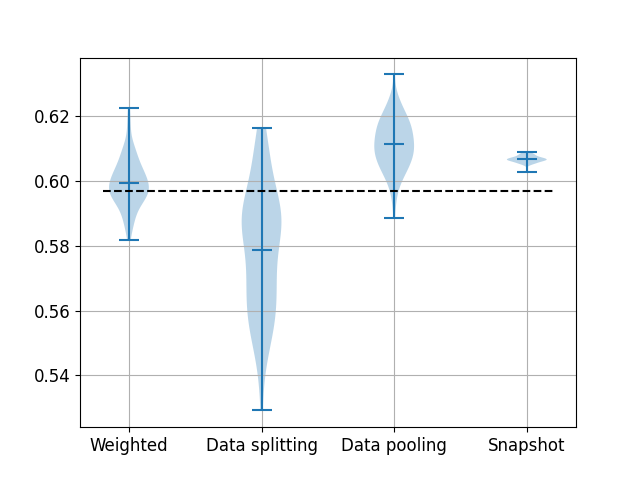}}
\subfigure[treatment: SD
	]{ \includegraphics[width=2.1in]{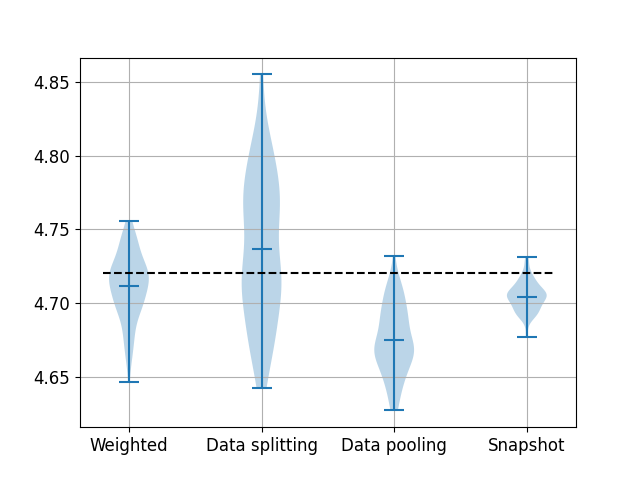}}
\subfigure[treatment: FR]{
\includegraphics[width=2.1in]{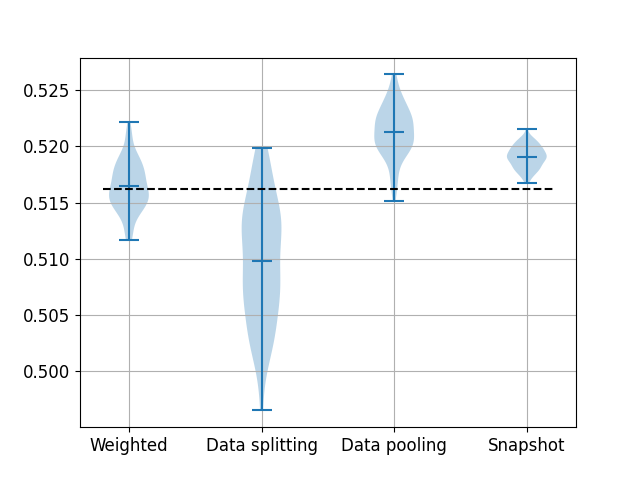}}
\subfigure[control: proportion]{
\includegraphics[width=2.1in]{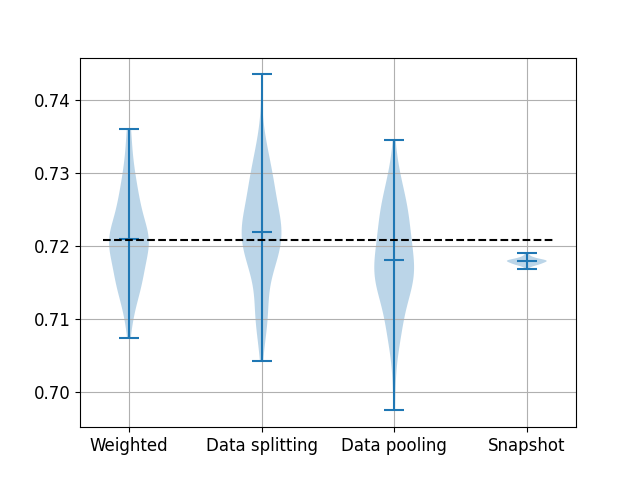}}
\subfigure[control: SD]{
\includegraphics[width=2.1in]{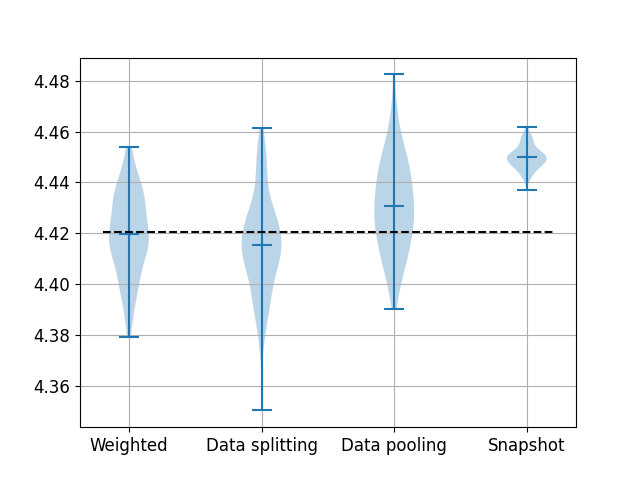}}
\subfigure[control: FR]{
\includegraphics[width=2.1in]{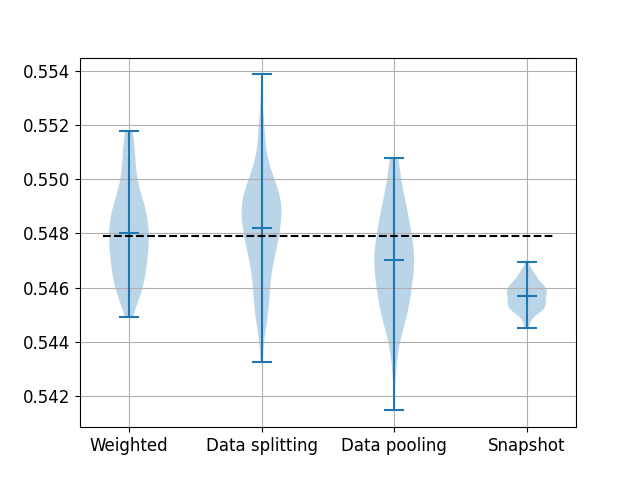}}
\caption{A/B testing results for $\protect\alpha_C=10$, $\protect\alpha_T=9$%
, and $p =0.2$ }
\label{fig:ab_p=0.2}
\end{figure}
\begin{table}[!htb]
\caption{Bias, standard deviation, and standard error estimated from the
experiment for the metrics in the case that $\protect\alpha_C=10$, $\protect%
\alpha_T=9$, and $p = 0.2$}
\label{tab:bias_AB_0.2}\centering
\begin{tabular}{lccccccccc}
\toprule & \multicolumn{3}{c}{Proportion of short videos} & 
\multicolumn{3}{c}{Stay durations} & \multicolumn{3}{c}{Finishing rates}  \\ 
\midrule & Bias & STD & SE & Bias & STD & SE & Bias & STD & SE \\ 
Weighted & 0.002 & 0.004 & 0.001 & -0.008 & 0.014 & 0.011 & 0.000 & 0.002 & 
0.001 \\ 
Data splitting & -0.019 & 0.020 & 0.001 & 0.021 & 0.052 & 0.011 & -0.007 & 
0.006 & 0.001 \\ 
Data pooling & 0.017 & 0.002 & 0.001 & -0.056 & 0.012 & 0.011 & 0.006 & 0.001
& 0.001 \\ 
Snapshot & 0.013 & 0.001 & 0.001 & -0.046 & 0.011 & 0.011 & 0.005 & 0.001 & 
0.001 \\ 
\bottomrule 
\end{tabular}%
\end{table}

\begin{table}[!ht]
\caption{Experimentation values in the case that $\protect\alpha_C=10$, $%
\protect\alpha_T=9$, and $p = 0.2$}
\label{tab:cost_AB_0.2}\centering
\begin{tabular}{lcc}
\toprule & Treatment values & Control values \\ 
\midrule Global & 9.8823 $\pm$ 0.0006 & 9.3515 $\pm$ 0.0005 \\ 
Weighted & 9.8757 $\pm$ 0.0013 & 9.3517 $\pm$ 0.0006 \\ 
Data splitting & 9.8347 $\pm$ 0.0015 & 9.3492 $\pm$ 0.0007 \\ 
Data pooling & 9.8877 $\pm$ 0.0013 & 9.3538 $\pm$ 0.0006 \\ 
Snapshot & 9.8949 $\pm$ 0.0015 & 9.3611 $\pm$ 0.0006 \\ 
\bottomrule
\end{tabular}%
\end{table}
\newpage
\subsection{A/A Tests}
In this section, we have conducted simulations for A/A tests, specifically choosing parameters such as $\alpha_C=\alpha_T=10$ with a treatment assignment probability of $p=1/2$. Since the treatment and control groups share an identical parameter, the global treatment effects should ideally be zero. In Figure \ref{fig:aa_p=1/2}, we present visualizations of treatment effect estimations for four methods. Notably, the weighted training, data pooling, and snapshot methods exhibit similar performance. Table \ref{tab:AAtestp=1/2} offers details on the average estimations and type I errors obtained from various methods, gathered from 100 independent runs of the A/A tests, with a confidence level set at $0.95$.
\begin{figure}[!ht]
	\centering 
	\subfigure[treatment effects: proportion]{
		\includegraphics[width=2.1in]{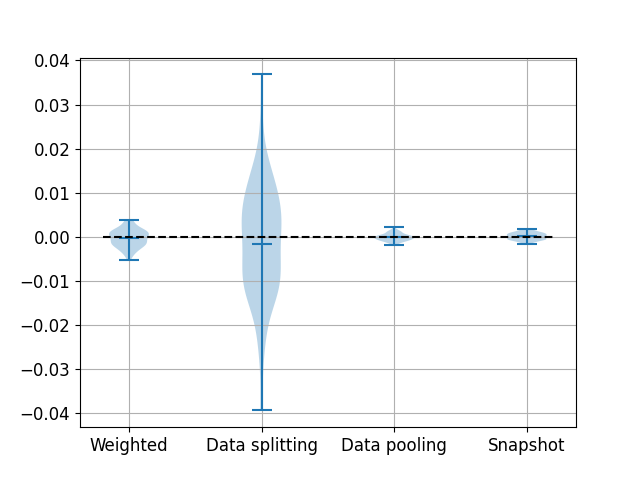}}
	\subfigure[treatment effects: SD]{
		\includegraphics[width=2.1in]{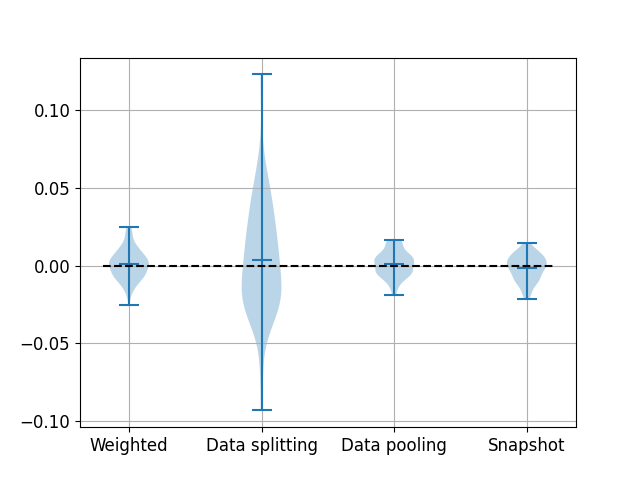}}
	\subfigure[treatment effects: FR]{
		\includegraphics[width=2.1in]{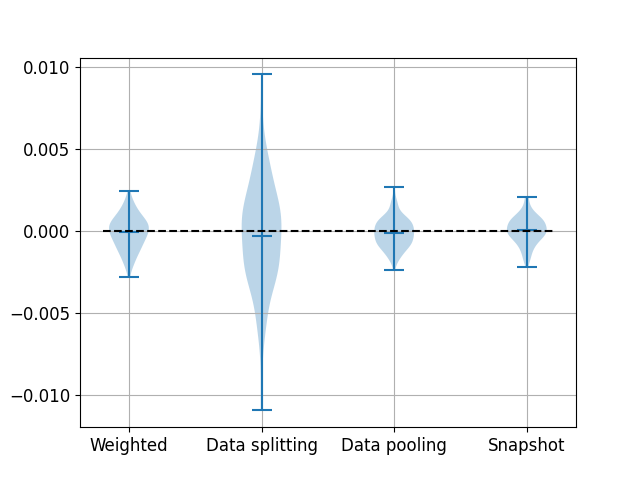}}
	\caption{A/A testing results for $\protect\alpha_C=\alpha_T=10$
		and $p =1/2$}
	\label{fig:aa_p=1/2}
\end{figure}
\begin{table}[!htb]
	\caption{The average estimations and type I error for the A/A test with $\alpha_C=\alpha_T=10$ and $p=1/2$}
	\label{tab:AAtestp=1/2}\centering
	\begin{tabular}{lcccccc}
		\toprule & \multicolumn{2}{c}{Proportion of short videos} & 
		\multicolumn{2}{c}{Stay durations} & \multicolumn{2}{c}{Finishing rates} \\ 
		\midrule & Estimation & Type I error & Estimation & Type I error & Estimation
		& Type I error \\ 
		Weighted & -0.0003 & 0.45  & 0.0008 & 0.09  & -0.0001 & 0.11 \\
		Data splitting & -0.0017 & 0.94  & 0.0039 & 0.65  & -0.0003 & 0.60 \\
		Data pooling & -0.0001 & 0.04  & 0.0011 & 0.07  & -0.0001 & 0.07 \\
		Snapshot & 0.0001 & 0.06  & -0.0015 & 0.06  & 0.0000 & 0.06 \\ 
		\bottomrule
	\end{tabular}%
\end{table}

It's noteworthy that our approach exhibits a slightly larger type I error than the target of 0.05 for the metrics stay durations (SD) and finishing rates (FR), and it demonstrates a worse type I error for the metric proportion of short videos. We attribute this behavior to the sensitivity of the proportion of short videos metric to the starting period of the experiment, which may be more feedback-loop dependent.

On the contrary, the data splitting method yields much higher Type I errors, suggesting that new inference methods should be developed to address this issue.
\label{subsection:AA}

% Table generated by Excel2LaTeX from sheet '10VS10 0.5'

\section{Concluding Remarks}
\label{sec:conclusion}
In this paper, we have introduced a weighted training approach designed to address the interference problem caused by data training loops. Our approach has demonstrated the capability to achieve low bias and reasonable variance. For future research, we have identified several intriguing directions:
\begin{enumerate}
\item \textbf{Single model training:} In our current approach, we still require training two separate models, which can be computationally expensive, especially when dealing with large machine learning models. It would be interesting to explore whether it's possible to train a single model and implement adjustments to mitigate bias effectively. This could lead to more efficient and practical solutions.

\item \textbf{Variance estimation and new inference methods:} Although our approach has shown promise in reducing bias, the variance remains larger than the standard error estimated from the two-sample t-test in some cases. As a result, there is a need for more robust methods for estimating variance and developing new inference techniques that can account for the specific challenges in interference induced by data training loops in A/B tests.
\end{enumerate}
Exploring these directions could further enhance our understanding of this type of interference and lead to more effective and efficient solutions for mitigating its biases.

\section*{Acknowledgement}
We would like to thank Jose Blanchet, Ramesh Johari, Shuangning Li, Zikun Ye, and Xinyu Yue for helpful discussions.

\bibliographystyle{plainnat}
\bibliography{mybib,feedback}
\newpage 
\begin{appendices} \section{Proofs} \label{appendix:proof} 
	\begin{proof}[Proof of Lemma \ref{lma:equal}]
		By the causal graph (Figure \ref{figure:causal}), we have $%
		Y_{E}\bot Z|X_{E},$ which yields%
		\begin{equation*}
			\mathbb{E}\left[ Z|X_{E}\right] =\mathbb{E}\left[ Z|X_{E},Y_{E}\right] =%
			\mathbb{E}\left[ Z|D_{E}\right] .
		\end{equation*}%
		For any measurable set $A\subset \mathbb{R}^d\times  \mathbb{R}^m$, we have 
		\begin{equation*}
			W_{T}\mathcal{D}_{E}(A){=\frac{1}{p}}\mathbb{E}\left[ \mathbb{E}\left[
			Z|D_{E}\right] I\left\{ D_{E}\in A\right\} \right] 
		\end{equation*}%
		Due to the property of the conditional expectation \citep[Theorem 5.1.7]{durrett2019probability}, we have
		\begin{equation*}
			\mathbb{E}\left[ \mathbb{E}\left[ Z|D_{E}\right] I\left\{ D_{E}\in A\right\} %
			\right] =\mathbb{E}\left[ \mathbb{E}\left[ I\left\{ D_{E}\in A\right\}
			Z|D_{E}\right] \right] =\mathbb{E}\left[ I\left\{ D_{E}\in A\right\} Z\right]
		\end{equation*}%
		Recall that $D_{E}=D_{T}Z+D_{C}\left( 1-Z\right) ,$ we have 
		\begin{equation*}
			\mathbb{E}\left[ Z I\left\{ D_{E}\in A\right\} \right] =\mathbb{E}%
			\left[ I\{Z=1\}I\left\{ D_{T}\in A\right\} \right] .
		\end{equation*}%
		Because of the independence of $Z$ and $D_{T},$ we have 
		\begin{equation*}
			\frac{1}{p}\mathbb{E}\left[ I\{Z=1\}I\left\{ D_{E}\in A\right\} \right] =%
			\mathbb{E}\left[ I\left\{ D_{T}\in A\right\} \right] .
		\end{equation*}.
	\end{proof}

\begin{proof}[Proof of Theorem \ref{thm:min_var}] We will focus on the treatment problem (\ref{problem:control_var}),
	as the control problem follows an identical approach. Recall that 
	\begin{equation*}
		W\mathcal{D}_{E}\mathcal{(}A\mathcal{)=}\mathbb{E}\left[ WI\left\{ D_{E}\in
		A\right\} \right] =\mathbb{E}\left[ \mathbb{E}\left[ W|D_{E}\right] I\left\{
		D_{E}\in A\right\} \right] 
	\end{equation*}
	for any measurable set $A$ in $\mathcal{X}\times \mathcal{Y}.$ On the other
	hand, we have 
	\begin{eqnarray*}
		p\mathbb{E}\left[ I\left\{ D_{T}\in A\right\} \right]  &=&\mathbb{E}\left[ Z%
		\right] \mathbb{E}\left[ I\left\{ D_{T}\in A\right\} \right] =\mathbb{E}%
		\left[ ZI\left\{ D_{T}\in A\right\} \right]  \\
		&=&\mathbb{E}\left[ ZI\left\{ D_{E}\in A\right\} \right] .
	\end{eqnarray*}%
	Therefore, the constraint means that 
	\begin{eqnarray*}
		\mathbb{E}\left[ \mathbb{E}\left[ W|D_{E}\right] I\left\{ D_{E}\in A\right\} %
		\right]  &=&\mathbb{E}\left[ I\left\{ D_{T}\in A\right\} \right]  \\
		&=&\mathbb{E}\left[ \left( Z/p\right) I\left\{ D_{E}\in A\right\} \right] ,
	\end{eqnarray*}%
	for any measurable set $A$ in $\mathcal{X}\times \mathcal{Y}.$ By the
	definition of the conditional expectation \citep[Section 5.1]{durrett2019probability}, we have 
	\begin{equation*}
		\mathbb{E}\left[ W|D_{E}\right] =\frac{1}{p}\mathbb{E}\left[ Z|D_{E}\right] .
	\end{equation*}%
	By Theorem 5.1.3 in \citet{durrett2019probability}, we have
	\begin{equation*}
		\mathbb{E}\left[ W^{2}\right] =\mathbb{E}\left[ \mathbb{E}\left[ W^{2}|D_{E}%
		\right] \right] \geq \mathbb{E}\left[ \left( \mathbb{E}\left[ W|D_{E}\right]
		\right) ^{2}\right] =\mathbb{E}\left[ \left( \mathbb{E}\left[ Z|D_{E}\right]
		/p\right) ^{2}\right] .
	\end{equation*}
We conclude the proof by noting that 
$$ \mathbb{E}\left[ Z|D_{E}\right] =\mathbb{E} \left[ Z|X_{E}\right],$$
as also shown in the proof of Lemma \ref{lma:equal}.
\end{proof}
\section{Additional Numerical Results}
\label{appendix:additional}
\subsection{A/B Tests}

In this subsection, we present additional A/B testing simulations. Firstly, we consider $\alpha_C=10$, $\alpha_T=8$, and $p = 1/2$, and the results are visualized in Figure \ref{fig:ab_p=1/2_10VS8}, while detailed bias, variance, and cost findings can be found in Tables \ref{tab:bias_AB_1/2_10_8} and \ref{tab:cost_AB_0.5_10_8}. Additionally, we explore the scenario with $\alpha_C=10$, $\alpha_T=9$, and $p = 0.3$, the results of which are illustrated in Figure \ref{fig:ab_p=0.3}, and detailed bias, variance and cost results can be found in Tables \ref{tab:bias_AB_0.3} and \ref{tab:cost_AB_0.3_10_9}. 
\begin{figure}[!ht]
	\centering 
	\subfigure[treatment effects: proportion]{
		\includegraphics[width=2.1in]{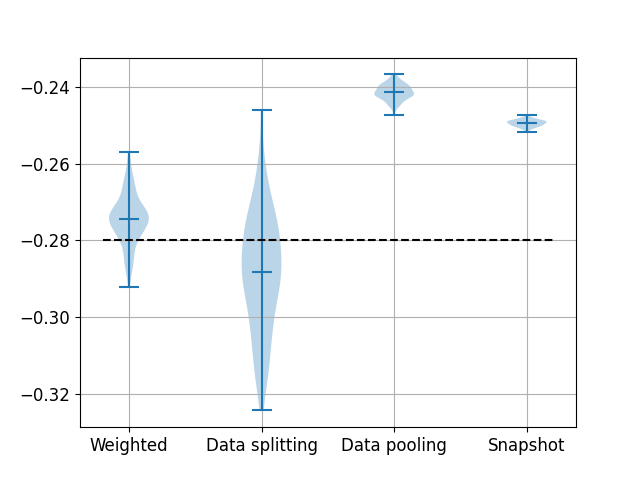}}
	\subfigure[treatment effects: SD]{
		\includegraphics[width=2.1in]{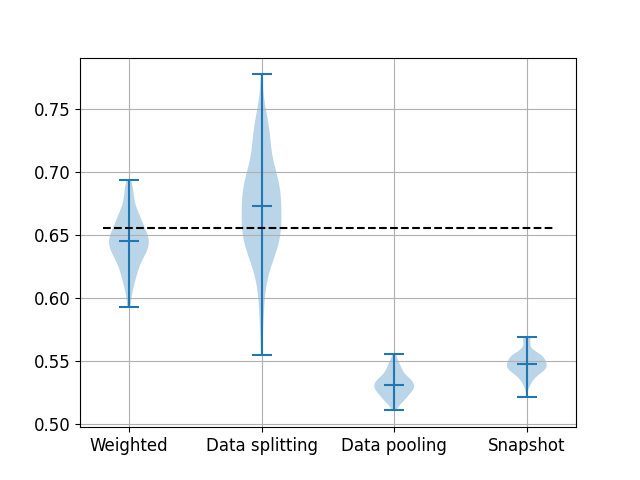}}
	\subfigure[treatment effects: FR]{
		\includegraphics[width=2.1in]{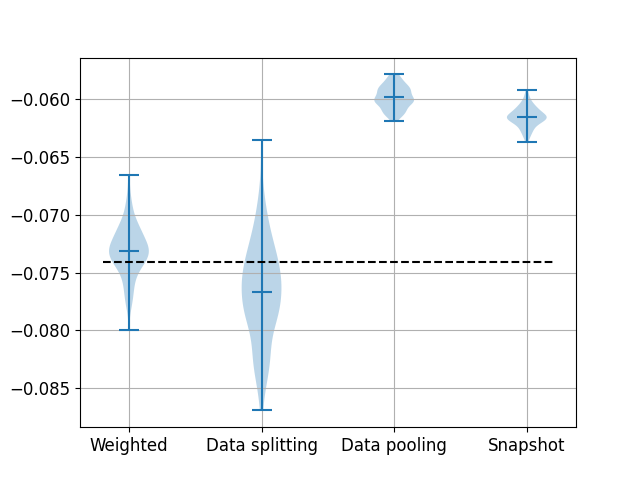}}
	\subfigure[treatment: proportion]{ 
		\includegraphics[width=2.1in]{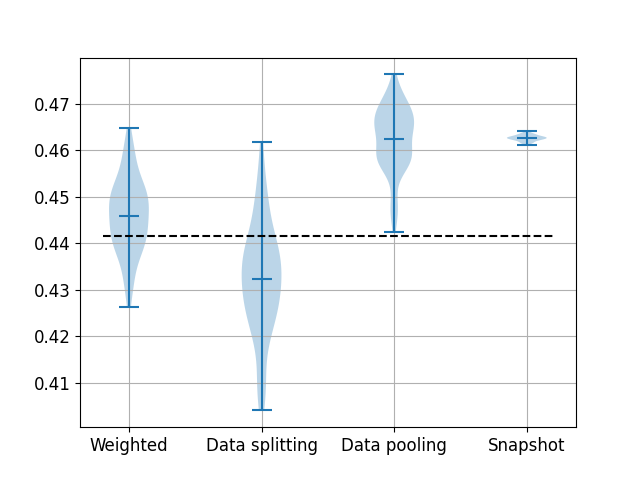}}
	\subfigure[treatment: SD
	]{ \includegraphics[width=2.1in]{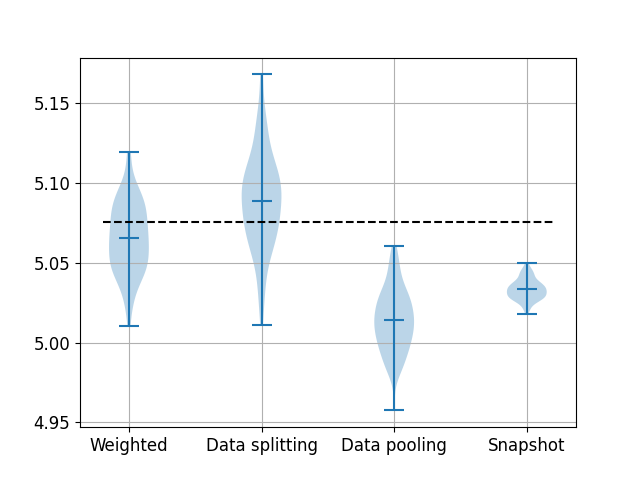}}
	\subfigure[treatment: FR]{
		\includegraphics[width=2.1in]{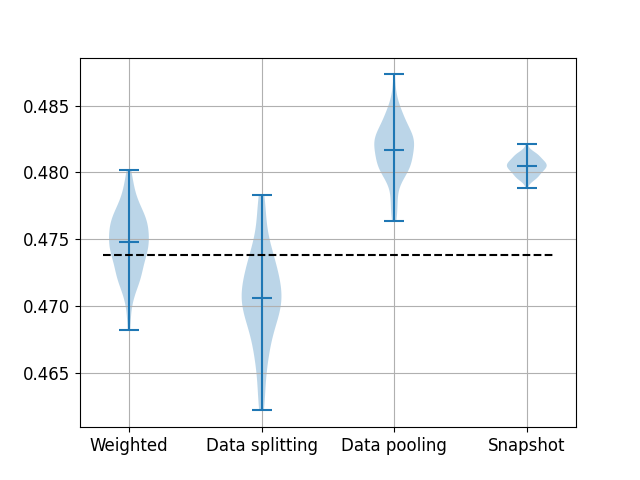}}
	\subfigure[control: proportion]{
		\includegraphics[width=2.1in]{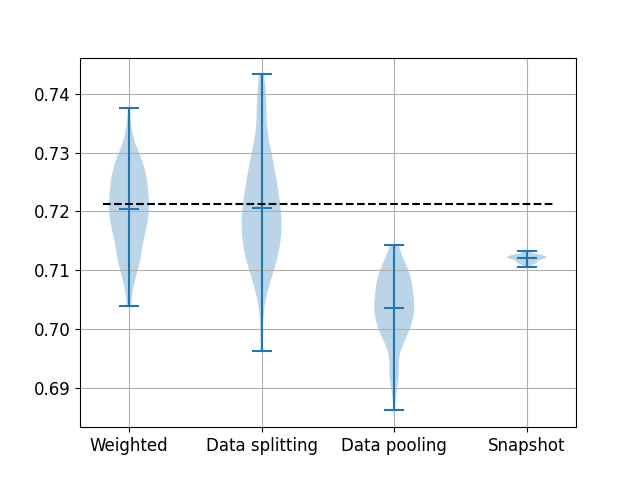}}
	\subfigure[control: SD]{
		\includegraphics[width=2.1in]{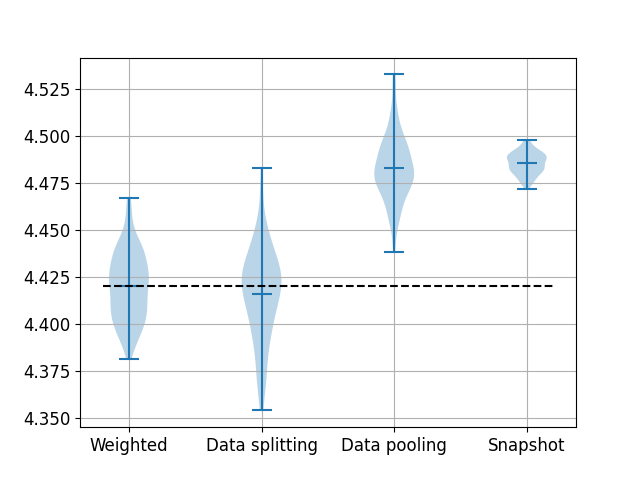}}
	\subfigure[control: FR]{
		\includegraphics[width=2.1in]{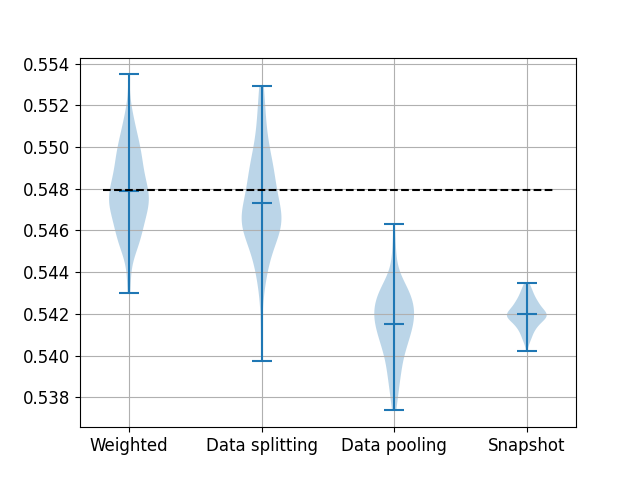}}
	\caption{A/B testing results for $\protect\alpha_C=10$, $\protect\alpha_T=8$%
		, and $p = 1/2$ }
	\label{fig:ab_p=1/2_10VS8}
\end{figure}
\begin{table}[!htb]
	\caption{Bias, standard deviation, and standard error estimated from the
		experiment for the metrics in the case that $\protect\alpha_C=10$, $\protect%
		\alpha_T=8$, and $p = 1/2$}
	\label{tab:bias_AB_1/2_10_8}\centering
	\begin{tabular}{lccccccccc}
		\toprule & \multicolumn{3}{c}{Proportion of short videos} & 
		\multicolumn{3}{c}{Stay durations} & \multicolumn{3}{c}{Finishing rates}  \\ 
		\midrule & Bias & STD & SE & Bias & STD & SE & Bias & STD & SE \\ 
		
		Weighted & 0.005 & 0.007 & 0.001 & -0.010 & 0.021 & 0.009 & 0.001 & 0.002 & 0.001 \\
		Data splitting & -0.008 & 0.015 & 0.001 & 0.018 & 0.040 & 0.009 & -0.003 & 0.004 & 0.001 \\
		Data pooling & 0.039 & 0.002 & 0.001 & -0.124 & 0.009 & 0.009 & 0.014 & 0.001 & 0.001 \\
		Snapshot & 0.031 & 0.001 & 0.001 & -0.107 & 0.009 & 0.009 & 0.013 & 0.001 & 0.001 \\
		\bottomrule
	\end{tabular}%
\end{table}
\begin{table}[!ht]
	\caption{Experimentation values in the case that $\protect\alpha_C=10$, $%
		\protect\alpha_T=8$, and $p = 0.5$}
	\label{tab:cost_AB_0.5_10_8}\centering
	\begin{tabular}{lcc}
		\toprule & Treatment values & Control values \\ 
		\midrule 
		Global & 9.8144 $\pm$ 0.0007 & 8.8040 $\pm$ 0.0007 \\
		Weighted & 9.8132 $\pm$ 0.0008 & 8.8032 $\pm$ 0.0008 \\
		Data splitting &  9.7953 $\pm$ 0.0009 & 8.7946 $\pm$ 0.0009 \\
		Data  pooling & 9.8312 $\pm$ 0.0007 & 8.8151 $\pm$ 0.0007 \\
		Snapshot & 9.8383 $\pm$ 0.0009 & 8.8215 $\pm$ 0.0008 \\
		\bottomrule
	\end{tabular}%
\end{table}
\begin{table}[!ht]
	\caption{Bias, standard deviation, and standard error estimated from the
		experiment for the metrics in the case that $\protect\alpha_C=10$, $\protect%
		\alpha_T=9$, and $p = 0.3$}
	\label{tab:bias_AB_0.3}\centering
	\begin{tabular}{lccccccccc}
		\toprule & \multicolumn{3}{c}{Proportion of short videos} & 
		\multicolumn{3}{c}{Stay durations} & \multicolumn{3}{c}{Finishing rates}  \\ 
		\midrule & Bias & STD & SE & Bias & STD & SE & Bias & STD & SE \\ 
		{Weighted} & 0.002 & 0.004 & 0.001 & -0.004 & 0.013 & 0.009 & 0.000 & 0.002 & 0.001 \\
		{Data splitting} & -0.010 & 0.016 & 0.001 & 0.009 & 0.042 & 0.009 & -0.003 & 0.005 & 0.001 \\
		{Data pooling} & 0.017 & 0.002 & 0.001 & -0.056 & 0.009 & 0.009 & 0.006 & 0.001 & 0.001 \\
		{Snapshot} & 0.010 & 0.001 & 0.001 & -0.043 & 0.009 & 0.009 & 0.005 & 0.001 & 0.001 \\
		\bottomrule
	\end{tabular}%
\end{table}
\begin{table}[!ht]
	\caption{Experimentation values in the case that $\protect\alpha_C=10$, $%
		\protect\alpha_T=9$, and $p = 0.3$}
	\label{tab:cost_AB_0.3_10_9}\centering
	\begin{tabular}{lcc}
		\toprule & Treatment values & Control values \\ 
		\midrule 
		Global & 9.8824 $\pm$ 0.0006 & 9.3521 $\pm$ 0.0007 \\
		{Weighted} & 9.8820 $\pm$ 0.0011 & 9.3521 $\pm$ 0.0007 \\
		{Data splitting} & 9.8537 $\pm$ 0.0010 & 9.3476 $\pm$ 0.0007 \\
		{Data pooling} & 9.8860 $\pm$ 0.0010 & 9.3531 $\pm$ 0.0006 \\
		{Snapshot} & 9.8927 $\pm$ 0.0011 & 9.3628 $\pm$ 0.0007 \\
		
		\bottomrule
	\end{tabular}%
\end{table}

\begin{figure}[!ht]

	\centering 
	\subfigure[treatment effects: proportion]{
		\includegraphics[width=2.1in]{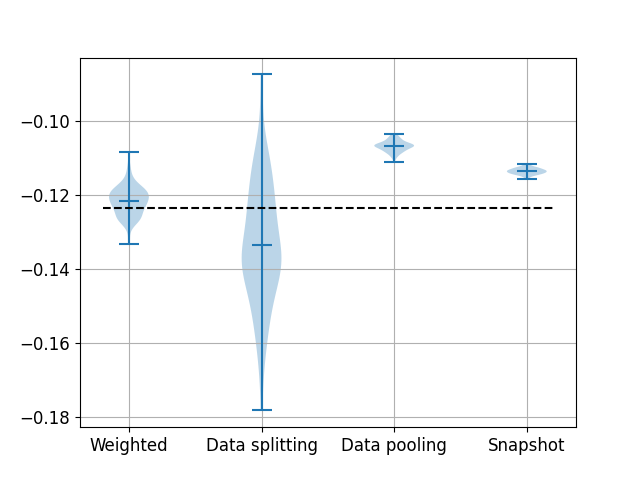}}
	\subfigure[treatment effects: SD]{
		\includegraphics[width=2.1in]{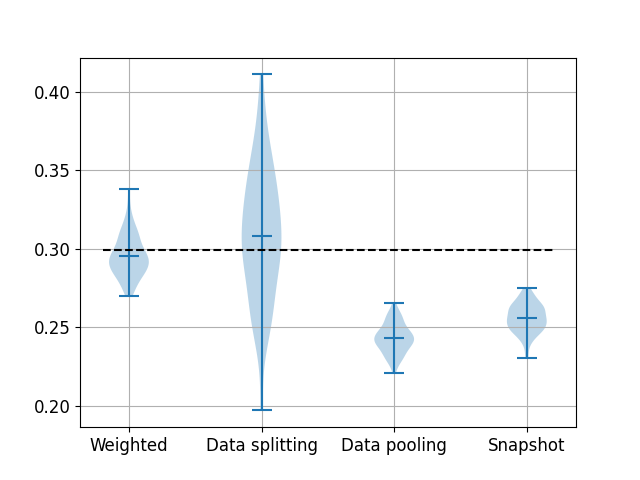}}
	\subfigure[treatment effects: FR]{
		\includegraphics[width=2.1in]{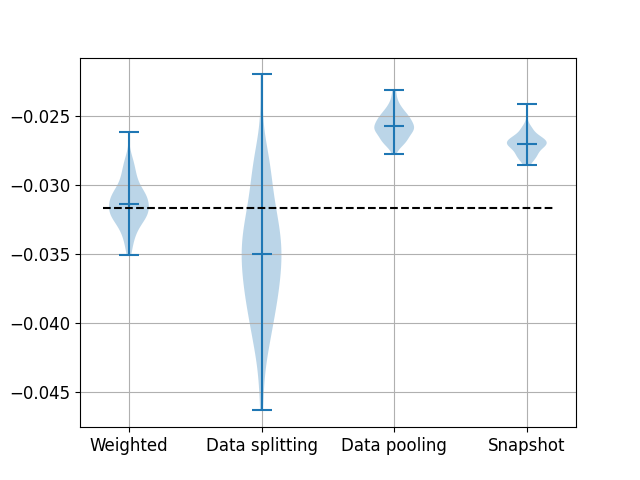}}
	\subfigure[treatment: proportion]{ 
		\includegraphics[width=2.1in]{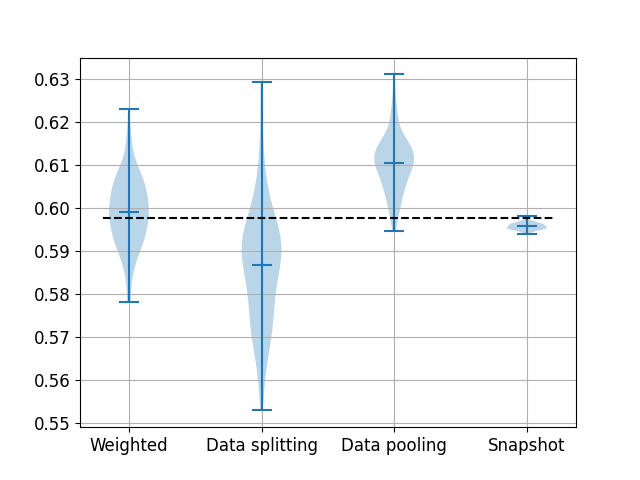}}
	\subfigure[treatment: SD
	]{ \includegraphics[width=2.1in]{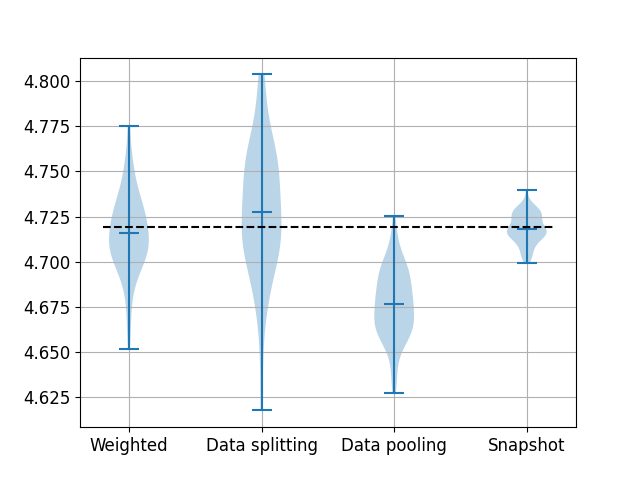}}
	\subfigure[treatment: FR]{
		\includegraphics[width=2.1in]{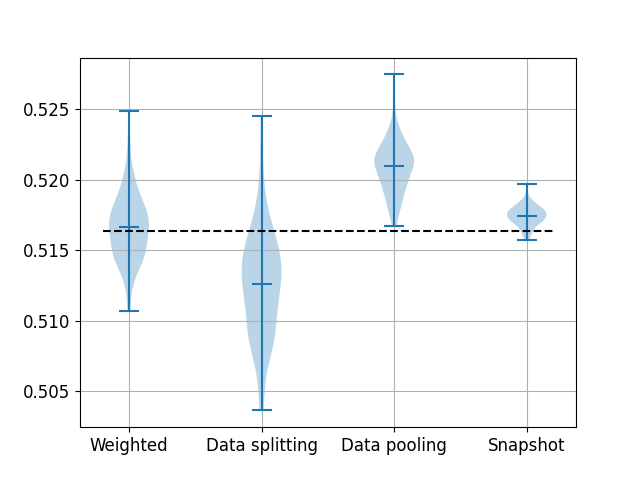}}
	\subfigure[control: proportion]{
		\includegraphics[width=2.1in]{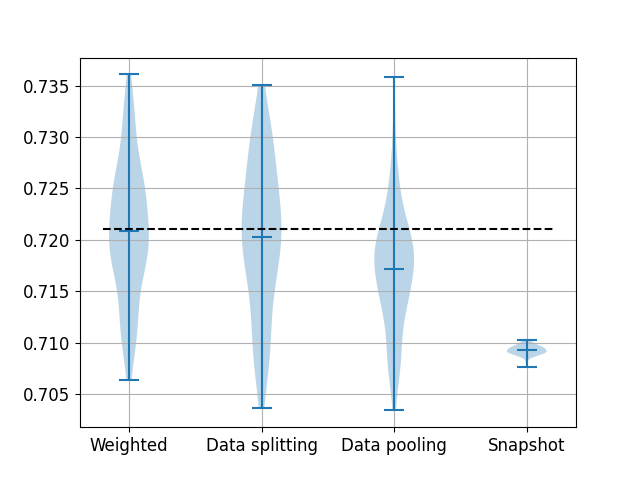}}
	\subfigure[control: SD]{
		\includegraphics[width=2.1in]{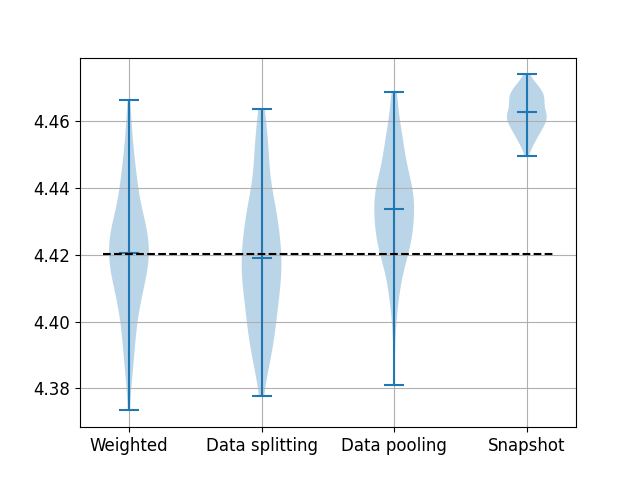}}
	\subfigure[control: FR]{
		\includegraphics[width=2.1in]{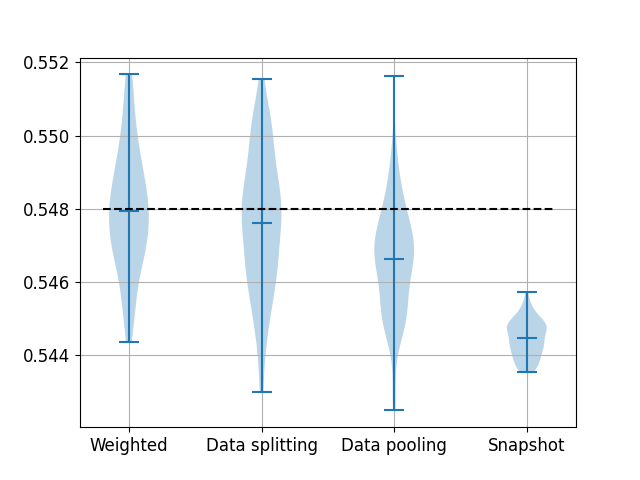}}
	\caption{A/B testing results for $\alpha_C=10$, $\alpha_T=9$%
		, and $p =0.3$ }
		\label{fig:ab_p=0.3}
	\end{figure}
\newpage 
\subsection{A/A Tests}
In this subsection, we present additional A/A testing results with $\alpha_C=\alpha_T=10$ and $p=0.2$. The estimations of treatment effects are visualized in Figure \ref{fig:aa_p=0.2}, and Table \ref{tab:AAtestp=0.2} offers comprehensive details regarding the estimators and type I errors.
\begin{table}[!htb]
	\caption{The average estimations and type I error for the A/A test with $\alpha_C=\alpha_T=10$ and $p=0.2$}
	\label{tab:AAtestp=0.2}\centering
	\begin{tabular}{lcccccc}
		\toprule & \multicolumn{2}{c}{Proportion of short videos} & 
		\multicolumn{2}{c}{Stay durations} & \multicolumn{2}{c}{Finishing rates} \\ 
		\midrule & Estimation & Type I error & Estimation & Type I error & Estimation
		& Type I error \\ 
		Weighted & -0.0004 & 0.47  & -0.0005 & 0.07  & -0.0002 & 0.08 \\
		Data splitting & -0.0073 & 0.91  & -0.0036 & 0.56  & -0.0031 & 0.75 \\
		Data pooling & 0.0001 & 0.04  & 0.0008 & 0.03  & 0.0000 & 0.05 \\
		Snapshot & -0.0001 & 0.07  & -0.0001 & 0.06  & 0.0001 & 0.04 \\
		
		\bottomrule
	\end{tabular}%
\end{table}
\begin{figure}[!ht]
	\centering 
	\subfigure[treatment effects: proportion]{
		\includegraphics[width=2.1in]{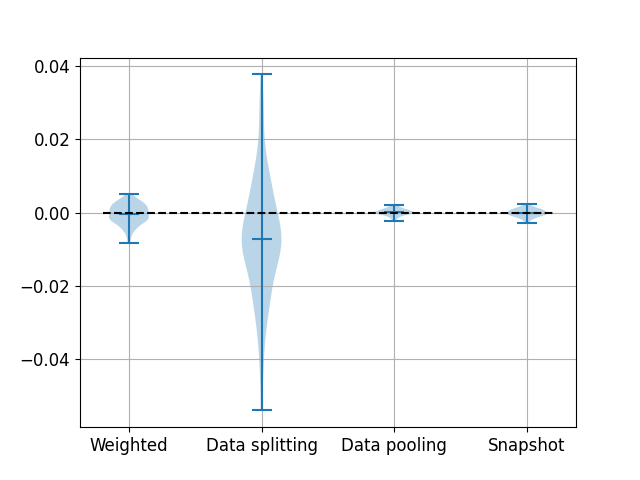}}
	\subfigure[treatment effects: SD]{
		\includegraphics[width=2.1in]{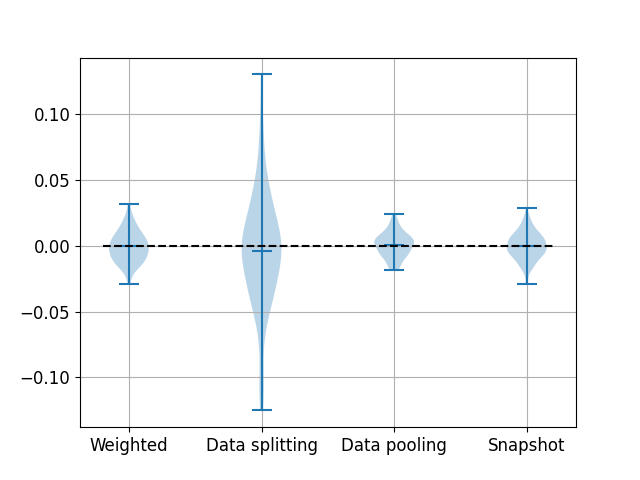}}
	\subfigure[treatment effects: FR]{
		\includegraphics[width=2.1in]{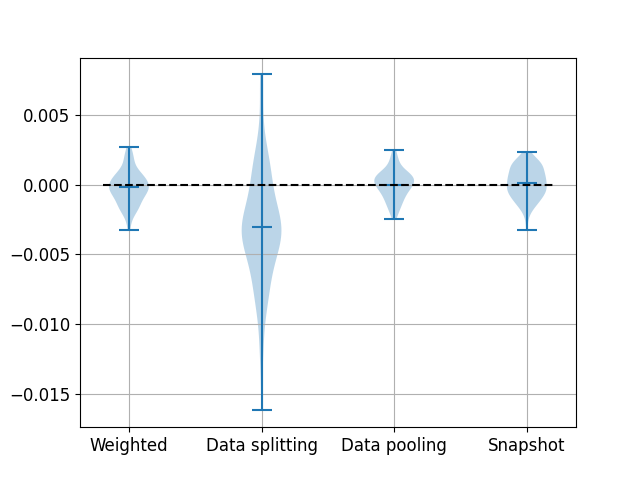}}
	\caption{A/A testing results for $\protect\alpha_C=\alpha_T=10$
		and $p =0.2$}
	\label{fig:aa_p=0.2}
\end{figure}

\end{appendices} 

\end{document}